\documentclass[10pt]{article}
    \usepackage[margin=2cm]{geometry}
    \usepackage{times,bbm,microtype}
    \usepackage{amsmath,amsthm,amssymb,stmaryrd,enumitem,x
    space,comment,booktabs}
    
    \usepackage[usenames, dvipsnames]{xcolor}
    \usepackage{algorithm}
    \usepackage[noend]{algpseudocode}
    \usepackage{thmtools,thm-restate}
    \usepackage{url}
    \usepackage[sort]{cite}
    \usepackage[linktocpage=true,
	pagebackref=true, colorlinks, linkcolor=BrickRed, citecolor=blue]{hyperref}
    \usepackage{cleveref}
    \usepackage{tikz} 
    \usepackage{caption}
    \title{New Algorithms and Lower Bounds for Streaming Tournaments}
    \author{Prantar Ghosh\footnote{(prantar.ghosh@gmail.com) Department of Computer Science, Georgetown University. Supported in part by NSF under award 1918989. Part of this work was done while the author was at DIMACS, Rutgers University, supported in part by a grant (820931) to DIMACS from the Simons Foundation.}
    \and 
    Sahil Kuchlous\footnote{(sahilkuchlous@college.harvard.edu) Harvard University. Research done in part as a participant in the DIMACS REU program 2023 at Rutgers University, supported by NSF grant CCF-2150186.}
    }

    \date{}
    
\newcommand{\xqedhere}[2]{\rlap{\hbox to#1{\hfil\llap{\ensuremath{#2}}}}}

\newcommand{\eps}{\varepsilon}

\newcommand{\mypar}[1]{\medskip\noindent{\bfseries #1.}~}

\newcommand{\bx}{\mathbf{x}}
\newcommand{\by}{\mathbf{y}}

\newcommand{\cA}{\mathcal{A}}

\newcommand{\cF}{\mathcal{F}}

\newcommand{\cO}{\mathcal{O}}

\newcommand{\cS}{\mathcal{S}}

\newcommand{\tO}{\widetilde{O}}
\newcommand{\tOmega}{\widetilde{\Omega}}
\newcommand{\tTheta}{\widetilde{\Theta}}

\renewcommand{\geq}{\geqslant}

\newcommand{\din}{d_{\text{in}}}
\newcommand{\dout}{d_{\text{out}}}
\newcommand{\scc}{\mathrm{SCC}}

\newcommand{\idx}{\textsc{index}\xspace}
\newcommand{\spidx}{\textsc{sp-index}\xspace}
\newcommand{\disj}{\textsc{disj}\xspace}
\newcommand{\tedge}{\textsc{t-edge}\xspace}
\newcommand{\fas}{\textsc{fas}\xspace}
\newcommand{\fast}{\textsc{fas-t}\xspace}
\newcommand{\fassiz}{\textsc{fas-size}\xspace}
\newcommand{\fassizt}{\textsc{fas-size-t}\xspace}

\newcommand{\hpath}{\textsc{ham-path}\xspace}
\newcommand{\hcyc}{\textsc{ham-cycle}\xspace}
\newcommand{\hpatht}{\textsc{ham-path-t}\xspace}
\newcommand{\hcyct}{\textsc{ham-cycle-t}\xspace}
\newcommand{\acyc}{\textsc{acyc}\xspace}
\newcommand{\sconn}{\textsc{str-conn}\xspace}
\newcommand{\sconnt}{\textsc{str-conn-t}\xspace}
\newcommand{\sccdect}{\textsc{scc-dec-t}\xspace}
\newcommand{\sccdec}{\textsc{scc-dec}\xspace}
\newcommand{\acyct}{\textsc{acyc-t}\xspace}
\newcommand{\sink}{\textsc{sink-dag}\xspace}
\newcommand{\sinkt}{\textsc{sink-dag-t}\xspace}
\newcommand{\setint}{\textsc{set-int}\xspace}
\newcommand{\reach}{\textsc{reach}\xspace}
\newcommand{\reacht}{\textsc{reach-t}\xspace}
\newcommand{\stdist}{\textsc{stdist}\xspace}
\newcommand{\stdistt}{\textsc{stdist-t}\xspace}
\newcommand{\uor}{\textsc{unique-or}\xspace}
\newcommand{\bzero}{\mathbf{0}}
\newcommand{\be}{\mathbf{e}}

\newcommand{\prantar}[1]{
}
\newcommand{\sahil}[1]{
}
\newcommand{\todo}[1]{
}

\newtheorem{theorem}{Theorem}[section]
\newtheorem{lemma}[theorem]{Lemma}

\newtheorem{corollary}[theorem]{Corollary}

\newtheorem{fact}[theorem]{Fact}

\newtheorem{remark}[theorem]{Remark}
\theoremstyle{definition}  \newtheorem{definition}[theorem]{Definition}

\algrenewcomment[1]{\hfill \textcolor{blue}{$\triangleright$ #1}}
\makeatletter
\algnewcommand{\LineComment}[1]{\Statex \hskip\ALG@thistlm \textcolor{blue}{$\triangleright$ #1}}
\algnewcommand{\LineCommentI}[1]{\Statex \hskip\ALG@thistlm \hskip\algorithmicindent \textcolor{blue}{$\triangleright$ #1}}

\begin{document}

   \maketitle \pagenumbering{roman}
   
    \begin{abstract}
    We study fundamental directed graph (digraph)
    problems in the streaming model. An initial 
    investigation by Chakrabarti, Ghosh, McGregor, and 
    Vorotnikova~[SODA'20] on streaming digraphs showed 
    that while most of these problems are provably hard 
    in general, some of them become tractable when 
    restricted to the well-studied class of 
    \emph{tournament} graphs where every pair of nodes 
    shares exactly one directed edge. Thus, we focus on 
    tournaments and improve the state of the art for 
    multiple problems in terms of both upper and lower bounds. 
    
    Our primary upper bound is a deterministic 
    single-pass semi-streaming algorithm (using $\tO(n)$ 
    space for $n$-node graphs, where $\tO(.)$ hides 
    polylog$(n)$ factors) for decomposing a tournament into 
    \emph{strongly connected components} (SCC). It is significant 
    from various angles: (i) it improves upon the previously best-known 
    algorithm by Baweja, Jia, and Woodruff [ITCS'22] in 
    terms of both space and passes: for $p\geq 1$, they 
    used $(p+1)$-passes and $\tO(n^{1+1/p})$-space, (ii) 
    it settles the single-pass streaming 
    complexities of fundamental problems like {\em strong 
    connectivity} and {\em reachability} on tournaments 
    as we show matching lower bounds for them, even for 
    randomized algorithms, and (iii) we use it as a 
    subroutine to obtain the first streaming algorithm 
    for \emph{Hamiltonian cycle} and improved algorithms 
    for \emph{Hamiltonian path} and \emph{feedback arc 
    set}. Further, we generalize the technique to obtain 
    \emph{optimal} single-pass algorithms for SCC decomposition 
    of a broader class of digraphs that are ``close'' to tournaments. 
    Our tight bounds demonstrate that the problem's complexity grows 
    smoothly with the ``distance'' from tournaments.

    While tournaments were indeed known to admit 
    efficient streaming algorithms for several problems, 
    there was a lack of \emph{strong} lower bounds for 
    these graphs. We prove the first $\Omega(n^2)$-space 
    lower bounds for this class, exhibiting that some 
    well-studied problems---such as (exact) 
    \emph{feedback arc set on tournaments} (\fast) and 
    \emph{$s,t$-distance}---remain hard here. In fact, we 
    obtain a stronger generalized lower bound on 
    space-approximation tradeoffs for \fast: any single-pass 
    $(1\pm \eps)$-approximation algorithm requires 
    $\Omega(n/\sqrt{\eps})$ space.  While the said work of 
    Chakrabarti et al.~gave an $\tO(n/\eps^2)$-space upper bound 
    for this problem, no non-trivial lower bound was known. Finally, 
    we settle the streaming complexities of two basic digraph problems 
    studied by prior work: \emph{acyclicity testing} of tournaments 
    and \emph{sink finding} in DAGs. As a whole, our collection of 
    results contributes significantly to the growing literature on streaming digraphs.
    \end{abstract}

    \newpage
    \tableofcontents

    \newpage
    \pagenumbering{arabic}

    \section{Introduction}\label{sec:intro}

A graph streaming algorithm reads an input graph by
making one or a few passes over a sequence of its 
edges, while maintaining a small summary that it uses to solve an underlying problem. While any 
memory sublinear in the number of edges is 
interesting, we typically aim for \emph{semi-streaming} space, i.e., $\tO(n)$ 
space\footnote{Throughout the paper, we use 
$\tO(f)$ to hide polylog$(f)$ factors.} for $n$-node graphs: Feigenbaum et 
al.~\cite{feigenbaumkmsz05} observed that many 
graph problems become tractable at this bound. The 
study of such algorithms is motivated by modern 
large graphs such as web graphs given by 
hyperlinks, social networks given by ``follows,'' 
citation networks, and biological networks. 
Chakrabarti, Ghosh, McGregor, and Vorotnikova 
\cite{ChakrabartiGMV20} observed that while many of 
these graphs are \emph{directed}, the graph 
streaming literature spanning over two decades had 
hitherto focused almost exclusively on undirected 
graphs (see \cite{McGregor14} for a survey), with 
very few exceptions. In light of this, they 
conducted a thorough investigation on streaming 
directed graphs (digraphs) and laid a foundation for their study in this model. 

Their findings, however, showed that most fundamental digraph problems are, in general, provably hard in streaming, perhaps justifying the lack of focus on these problems in prior work. At the same time, they found that restriction to \emph{tournament} graphs---where every pair of vertices shares exactly one directed edge---make many of these problems tractable yet non-trivial in streaming.  Observe that tournaments, by definition, are dense graphs with $\Theta(n^2)$ edges, and hence, one of the prime examples of graphs that motivate memory-efficient streaming algorithms. Furthermore, real-world graphs, e.g., ones representing comparison matrices that record pairwise preferences among a large set of items, are tournaments \cite{GassTournaments} and seek succinct sketches in the modern world of big data. Thus, the study of tournaments is well-motivated in the streaming model and in this work, we significantly contribute to this study.

Chakrabarti et al.~\cite{ChakrabartiGMV20} studied classical digraph problems including 
$s,t$-\textit{reachability}, \textit{topological sorting}, \textit{acyclicity testing}, and 
(approximate) \textit{feedback arc set}. These problems turn out to be hard for general digraphs: they need $\Omega(n^2)$ space in a single pass 
\cite{feigenbaumkmsz05, ChakrabartiGMV20} (i.e., we 
need to store almost the entire graph) and roughly 
$n^{1+\Omega(1/p)}$ space for $p$ passes 
\cite{GuruswamiO13, ChakrabartiGMV20}. 
Further, \cite{ChakrabartiGMV20} extended the multipass 
lower bounds to random-order streams. In contrast, they showed that the last three problems in the above list
admit single-pass semi-streaming algorithms on tournaments. However, two important questions remained: (i) ``\emph{what is the streaming complexity of $s,t$-reachability on tournaments?}'' and (ii) ``{\em which problems are hard on streaming tournaments?}'' 

Reachability on general digraphs, one of the most fundamental graph algorithmic problems, has received considerable attention in the streaming literature. It has strong lower bounds establishing $\Omega(n^{2-o(1)})$ space to be necessary even for $O(\sqrt{\log n})$ passes \cite{ChenKPSSY21-lognpass} (also see \cite{AssadiR20} who proved it for $2$ passes), while almost $\log n$ passes are needed for semi-streaming space \cite{GuruswamiO13}. On the upper-bound side, Assadi et al.~\cite{AssadiJJST22} designed an $O(\sqrt{n})$-pass semi-streaming algorithm. Closing this gap between the number of passes needed for semi-streaming space seems difficult and it is one of the major open questions in multipass streaming (see the very recent survey by Assadi \cite{AssadiMultipassSurvey}). Thus, resolving the complexity of the problem for the special class of tournaments seems to be a natural basic step before we try it for the general case. 

 Baweja, Jia, and Woodruff \cite{BawejaJW22} gave the first algorithm for reachability on tournaments: although they didn't explicitly mention it, a $(p+1)$-pass $\tO(n^{1+1/p})$-space algorithm follows from their SCC (strongly connected components) decomposition algorithm for tournaments that uses the same space-pass tradeoff. But this bound was not proven to be tight for either reachability or SCC decomposition,\footnote{For general digraphs, \cite{BawejaJW22} proved an $n^{1+\Omega(1/p)}$ space lower bound for $p$-pass algorithms solving SCC decomposition. This, however, doesn't extend to tournaments, and hence, doesn't complement the said upper bound.} and hence the former's streaming complexity on tournaments still remained unresolved. Additionally, one can ask more generally, ``\emph{what is the streaming complexity of SCC decomposition on tournaments?}''   

In this work, we answer this general question by designing a deterministic single-pass semi-streaming algorithm for SCC decomposition on tournaments. In fact, our algorithm works for a broader class of digraphs: those with \emph{at least} one directed edge between each pair of vertices. Our algorithm is very simple: during the stream we only need to store the degrees of each vertex! The post-processing phase and its analysis are more elaborate, and we demonstrate how we can derive the \emph{SCC-graph} (obtained by contracting each SCC into a supernode) from just the degree information.\footnote{We were informed via personal communication that it was known how to test strong connectivity using the degree information as it appeared as an IOITC (International Olympiad in Informatics Training Camp) problem, but we did not find any published result or documented version of it. Nevertheless, finding the SCC-graph requires considerably more work than just testing strong connectivity.} Consequently, we get single-pass semi-streaming algorithms for reachability and strong connectivity on tournaments. We prove matching lower bounds for each of these problems, thus settling their complexity as well.

We extend our algorithms to handle graphs that are ``almost'' tournaments. Precisely, for digraphs that need to add or delete a total of $k$ edges to form a tournament (i.e., are \emph{$k$-close} to tournament), we obtain $\tO(n+k)$-space solutions by using our tournament algorithms as subroutines. Additionally, our matching $\Omega(n+k)$-space lower bounds for these graphs show that this is as good as generalizations can get: it exhibits how the complexity of the problem increases smoothly with the input's ``distance'' from the tournament property and highlights the importance of the property in attaining sublinear-space solutions. 
We show further applications of our SCC decomposition framework for tournaments: it yields efficient streaming algorithms for problems like \emph{Hamiltonian cycle}, \emph{Hamiltonian path}, and {\em feedback arc set} on tournaments (see \Cref{sec:results} for details). 

As we continue to find efficient tournament algorithms for well-studied problems, it is natural to revisit the second important question mentioned above: which streaming problems remain hard on this class of graphs? Prior work did not show any strong lower bound on these graphs; in particular, we did not know any tournament problem with a single-pass $\Omega(n^2)$-space lower bound. We prove the first such lower bound for two well-studied problems, namely $s,t$-\emph{distance} and (exact) \emph{feedback arc set} on tournaments (\fast), and extend the latter result to a generalized lower bound for space-approximation tradeoffs. We prove this general result using communication complexity, as is standard for streaming lower bounds, but the underlying communication lower bound uses a stronger version of the \emph{direct-sum} argument \cite{chakrabartiswy01} that might be of independent interest.

Finally, we settle the streaming complexities of two basic problems studied by prior work: acyclicity testing of tournaments and sink finding in DAGs. Collectively, our results significantly advance the state of the art of streaming digraph problems.

\subsection{Our Contributions and Comparison to Prior Work}\label{sec:results}

\begin{table*}[!hbt]
\begin{minipage}{\textwidth} % used for centering table
\centering
\begin{tabular}{c c c c c}
\toprule
{\bf Problem}
& {\bf Passes}  
& {\bf Space} 
& {\bf Reference}
& {\bf Notes}
\\
\midrule
\sccdect
& $1$
& $\Theta(n\log n)$
& \Cref{cor:sccdect}
& lower bound due to output size
\\
\sccdect 
& $p+1$ 
& $\tO(n^{1+1/p})$
& \cite{BawejaJW22}
& $p\geq 1$\\
 \sccdec  
 & $1$
 & $\tTheta(n+k)$
 & \Cref{cor:scc-gen-ub,cor:scc-gen-lb}
 & $k$ = closeness to tournament
\vspace{0.5em}
\\
\reacht, \sconnt
& $1$
& $\tO(n)$
& \Cref{cor:reacht,cor:strongconnt}
& \\
\reacht, \sconnt
& $p$
& $\Omega(n/p)$
& \Cref{thm:reach-lb,thm:strongconn-lb}
& $p\geq 1$
\\
\reacht, \sconnt  
& $p+1$ 
& $\tO(n^{1+1/p})$
& \cite{BawejaJW22}
& $p\geq 1$; \reacht result implicit\\
\reach, \sconn 
 & $1$
 & $\tTheta(n+k)$
 & \Cref{thm:reach-gen-ub,thm:reach-gen-lb}
 & $k$ = closeness to tournament
\vspace{0.5em}
\\
\hcyct
& $O(\log n)$
& $\tO(n)$
& \Cref{thm:hamcyc}
\vspace{0.5em}
\\
\fast
& $1$
& $\Omega(n^2)$
& \Cref{lem:exactfast}\\
$(1+\eps)$-approx.~\fast
& $1$
& $\Omega\left(\min\{n^2,n/\sqrt{\eps}\}\right)$
& \Cref{thm:fas-main}
& any $\eps>0$\\
$(1+\eps)$-approx.~\fast
& $1$
& $\tO\left(\min\{n^2,n/\eps^2\}\right)$
& \cite{ChakrabartiGMV20}
& any $\eps>0$
\vspace{0.5em}
\\
\stdistt
& $1$
& $\Omega(n^2)$
& \Cref{thm:stdist}
\vspace{0.5em}
\\
\acyct
& $p$
& $\tO(n/p)$
& \Cref{thm:acyc-t-ub}
& $p\geq 1$\\
\acyct
& $p$
& $\Omega(n/p)$
& \cite{ChakrabartiGMV20}, also \Cref{acyct-lbsimple}
& $p\geq 1$
\vspace{0.5em}
\\
\sink
& $p$
& $\Omega(n/p)$
& \Cref{thm:sinkfind-lb}
& $p\geq 1$\\
\sink
& $p$
& $O(n/p)$
& trivial
& $p\geq 1$\\
\bottomrule
\end{tabular}

\end{minipage}
\caption{%
Summary of our main upper and lower bounds, along with the state-of-the-art results for comparison. The problems are formally defined in \Cref{sec:results}. 
}
\label{table:results}
\end{table*}

\Cref{table:results} summarizes our main results. Below, we define each problem formally and discuss the context and comparisons to prior work. For each problem, the suffix ``\textsc{-t}'' indicates the version of the problem where the input is a tournament.

\mypar{SCC decomposition, strong connectivity, and reachability} Recall that a graph is called strongly connected if and only if each vertex is reachable from every other vertex. A strongly connected component (SCC) of a graph is a maximal subgraph that is strongly connected. In the streaming model, by SCC decomposition, we mean partitioning the vertex set into subsets where each subset induces an SCC (we don't need the edges contained inside the SCCs). Formally, the problem is defined as follows. 
\begin{description}
    \item \sccdec: Given a digraph $G=(V,E)$, output a partition $(V_1,\ldots,V_\ell)$ of $V$ such that each $V_i$ induces a strongly connected component of $G$. 
\end{description}
\noindent
We design a single-pass semi-streaming algorithm for \sccdect (\Cref{cor:sccdect}). Note that our algorithm not only outputs the partition, but the SCC-graph which is a DAG obtained by
contracting each SCC into a (super)node. Observe that 
since the SCC-graph of a tournament is an acyclic 
tournament, it can be represented simply as an ordering 
of its SCC's, where the edges between the components are 
implicit: all of them go from left to right. Our algorithm outputs this ordering. 

As a consequence, we get single-pass semi-streaming algorithms for checking reachability (\Cref{cor:reacht}) and strong 
 connectivity (\Cref{cor:strongconnt}) on tournaments. The general problems are defined below.

\begin{description}
    \item \sconn: Given a digraph $G=(V,E)$, output whether it is strongly connected.

    \item \reach: Given a digraph $G=(V,E)$ and nodes $s,t\in V$, output whether $t$ has a directed path from (i.e., is \emph{reachable} from) $s$. 
\end{description}
\noindent
Baweja, Jia, and Woodruff \cite{BawejaJW22} gave $\tO(n^{1+1/p})$-space algorithms for \sccdect and \sconnt using $p+1$ passes for any $p\geq 1$. An algorithm for \reacht using same space and passes is implicit, since it can be determined from the SCC-graph (which their algorithm also outputs). Observe that their algorithm needs almost $\log n$ passes to achieve semi-streaming space. Further, it needs at least $3$ passes to even attain sublinear ($o(n^2)$) space. In contrast, we achieve semi-streaming space in just a single pass. 

Further, \cite{BawejaJW22} showed an $n^{1+\Omega(1/p)}$ space lower bound for $p$-pass algorithms solving \sccdec or \sconn, while \cite{GuruswamiO13} gave a similar lower bound for \reach. Our results rule out the possibiity of extending these lower bounds to tournaments, and show a large gap between the complexity of each problem and its tournament version.

While semi-streaming space is clearly optimal for \sccdect since it is the space needed to simply present the output, it is not clear that the same holds for \sconnt and \reacht which have single-bit outputs. We prove matching $\Omega(n)$-space single-pass lower bounds for these problems (\Cref{thm:reach-lb,thm:strongconn-lb}). In fact, we prove more general lower bounds that show that semi-streaming space is optimal for them (up to polylogarithmic factors) even for polylog$(n)$ passes. It is important to note that while our upper bounds are deterministic, our lower bounds work even for randomized algorithms, thus completely settling the single-pass complexity of these problems (up to logarithmic factors). 

\mypar{Extension to almost-tournaments and tight bounds} A possible concern about our results
is that they are restricted to 
tournaments, a rather special class 
of digraphs. Indeed, our results are provably not generalizable to any arbitrary digraph since there exist graphs that don't admit sublinear-space solutions. Can we still cover a broader class of graphs? To address this, we first 
show that our algorithms work as is for any digraph without non-edges (hence allowing bidirected edges between arbitrary number of vertex pairs). Secondly, we generalize our algorithm to work for digraphs that are ``close to'' tournaments (but can have non-edges). A standard measure of closeness to a graph property $P$ (widely used in areas such as Property Testing) is the number of edges that need to be added or deleted such that the graph satisfies $P$. In similar vein, we define a digraph $G$ to be \emph{$k$-close} to a tournament if a total of at most $k$ edge additions and deletions to $G$ results in a tournament. We design an $\tO(n+k)$-space algorithm for any such digraph. Thirdly, and perhaps more importantly, we prove that this is tight: there exist digraphs $k$-close to tournaments where $\Omega(n+k)$ space is necessary. This exhibits a smooth transition in complexity based on the ``distance'' from the tournament property. We see this result as an important conceptual contribution of our work: ``tournament-like'' properties are indeed necessary to obtain efficient streaming algorithms for these problems and justifies the almost-exclusive focus of prior work \cite{ChakrabartiGMV20, BawejaJW22} as well as ours on the specific class of tournaments. 

Our algorithms for almost-tournaments use our scheme for (exact) tournaments as a subroutine. We call it $2^{\min(k,n)}$ times in the worst case, and we can do so without any error by leveraging the fact that our algorithm is deterministic. Although the streaming model doesn't focus on time complexity, the exponential runtime might be unsatisfactory. But we show that it is somewhat necessary, at least for algorithms that use a tournament algorithm as a blackbox (see \Cref{subsec:strconn-reach} for details). This is similar in spirit to a result of \cite{ChakrabartiGMV20} that justified the exponential runtime of their algorithm for approximate \fast. 

\mypar{Hamiltonian cycle and Hamiltonian path} We exhibit applications of our \sccdect algorithm to obtain efficient algorithms for the problems of finding Hamiltonian cycles and Hamiltonian paths.

\begin{description}
    \item \hcyc: Given a digraph $G=(V,E)$, find a directed cycle that contains all nodes in $V$ if one exists or output \textsc{none} otherwise.
    
    \item \hpath: Given a digraph $G=(V,E)$, find a directed path that contains all nodes in $V$ if one exists or output \textsc{none} otherwise.
\end{description}

\noindent
Note that although these problems are NP-hard in general \cite{Karp72}, they admit polynomial time solutions for tournaments \cite{beineke1978selected}. Further, it is known that every tournament contains a Hamiltonian path, but it may or may not contain a Hamiltonian cycle. A tournament is Hamiltonian if and only if it is strongly connected \cite{camion1959chemins}, and so we can use our \sconnt algorithm to determine the existence of a Hamiltonian cycle in a single-pass and semi-streaming space. Thus, the problem boils down to finding such a cycle if it exists. 

We give an $O(\log n)$-pass semi-streaming algorithm for \hcyct (\Cref{thm:hamcyc}). Although \cite{BawejaJW22} gave such an algorithm for \hpatht, to our knowledge, no streaming algorithm for \hcyct was known. We design our algorithm by using \sccdect as a primitive and devising a streaming simulation of a parallel algorithm by Soroker \cite{SorokerHamCycle}. Additionally, we prove that if one instead used \cite{BawejaJW22}'s \sccdect algorithm as the subroutine, it would lead to an $O(\log^2 n)$-pass algorithm in the worst case.

As another application of our SCC framework, we obtain improved algorithms for \hpatht when the input tournament is known to have somewhat small connected components. Observe that in the extreme case when the maximum size $s$ of an SCC is $1$, the tournament is a DAG, and a trivial algorithm that sorts the nodes by their indegrees finds the Hamiltonian path in a single semi-streaming pass. This is significantly better than the general $p$-pass $\tO(n^{1+1/p})$-space algorithm of \cite{BawejaJW22}. We demonstrate that even when $s$ is larger but still ``reasonably small'', specifically, whenever $s< n^{(p-1)/p}$ for some $p\geq 1$, we get better $p$-pass algorithms than \cite{BawejaJW22} and the space complexity grows smoothly with $s$. A similar algorithm for approximate \fast is also obtained. See \Cref{subsec:hamcyc} for details on these results. 

\mypar{Feedback arc set} The feedback arc set (FAS) in a digraph is a set of arcs (directed edges) whose deletion removes all directed cycles. Indeed, we are interested in finding the \emph{minimum} FAS. One version asks for only the \emph{size} of a minimum FAS.
\begin{description}
    \item \fassiz: Given a digraph $G=(V,E)$, output the size of a minimum FAS.
\end{description}
\noindent
For the version asking for the actual set, since the output can have size $\Theta(n^2)$ in the worst case, the streaming model considers an equivalent version where the output-size is always $\tO(n)$ \cite{ChakrabartiGMV20, BawejaJW22}. We define it below.

\begin{description}
    \item \fas: Given a digraph $G=(V,E)$, find an ordering of vertices in $V$ such that the number of back-edges (edge going from right to left in the ordering) is minimized. 
\end{description}
\noindent
The \fas and \fassiz problems, even on tournaments (\fast and 
\fassizt), are NP-hard \cite{CharbitTA07}. However, each of 
\fast and \fassizt admits a PTAS \cite{KenyonS07}. 
The \fas problem has a variety of applications including in rank aggregation \cite{KenyonS07}, machine learning \cite{BarYehudaGNR98}, social network analysis \cite{SimpsonST16}, and ranking algorithms \cite{GengLQALS08}. 
In the 
streaming model, \cite{ChakrabartiGMV20} and \cite{BawejaJW22} 
gave multipass polynomial-time approximation algorithms for 
these problems (see \Cref{sec:related} for details). We, 
however, focus on the single-pass setting. The best-known 
single-pass algorithm achieves $(1+\eps)$-approximation in 
only $\tO(n/\eps^2)$ space, albeit in exponential time 
\cite{ChakrabartiGMV20}. However, there was no complementary space 
lower bound for \fast or \fassizt. The said paper gave lower 
bounds only for the harder versions of the problems on general 
digraphs; these do not extend to tournaments. In fact, there 
was no lower bound even for \emph{exact} \fast. Note that 
although this is an NP-hard problem, an exponential-time but 
sublinear-space streaming algorithm is not ruled out. In fact, 
such algorithms for NP-hard problems do exist in the 
literature (see for instance, a deterministic $(1+\eps)$-approximation to correlation clustering (which is APX-hard)
\cite{AhnCGMW15, BehnezhadCMT23}, or algorithms for max 
coverage \cite{McGregorV17}). 

In light of this, and since 
the state-of-the-art $(1+\eps)$-approximation algorithm 
already uses exponential time, it is natural to ask whether it 
can be improved to solve exact \fast. 
We rule this out 
by proving a single-pass lower bound of $\Omega(n^2)$ space 
(\Cref{lem:exactfast}). No such lower bound was known for 
\emph{any} tournament problem prior to this.\footnote{It is, 
however, not hard to formulate a problem like $\tedge$ (see 
\Cref{sec:prelims}), for which such a lower bound is intuitive and easy to prove. Here, we mean that no ``textbook'' or well-studied digraph problem was known to have such a lower bound.} 
Additionally, we prove the same lower bound even for $\fassizt$ (\Cref{thm:fastsize-lb}) whose output is only a real 
value (rather than a vertex ordering like \fast).

We further extend our \fast lower bound to establish a smooth space-approximation tradeoff: given any $\eps>0$, a $(1+\eps)$ approximation to \fast requires $\Omega(\min\{n^2, n/\sqrt{\eps}\})$ space (\Cref{thm:fas-main}). Note that for $\eps\leq 1/n^2$, the lower bound follows directly from \Cref{lem:exactfast} since it is equivalent to the exact version. However, it does not follow that the bound holds for larger values of $\eps$. We prove that it is indeed the case for \emph{any} $\eps>0$.    

\mypar{Distance, acylicity testing, and sink-finding} It is natural to ask what other problems are hard on tournaments. In particular, is there any polynomial-time solvable well-studied problem that requires $\Omega(n^2)$ space in a single pass? We answer this in the affirmative by exhibiting the $s,t$-distance problem as an example.

\begin{description}
    \item \stdist: Given a digraph $G=(V,E)$ and vertices $s,t\in V$, find the distance, i.e., the length of the shortest directed path, between $s$ and $t$.
\end{description}
\noindent
We prove that $\Omega(n^2)$ space is necessary for \stdist in a single pass (\Cref{thm:stdist}). This also exhibits a contrast between \stdistt and \reacht in streaming: while checking whether $s$ has a path to $t$ is easy, finding the exact distance is hard.

Next, we revisit the basic problem of testing acylicity of a digraph studied by \cite{ChakrabartiGMV20, BawejaJW22}.
\begin{description}
    \item \acyc: Given a digraph $G=(V,E)$, is there a directed cycle in $G$? 
\end{description}
\noindent
We design a $p$-pass $\tO(n/p)$-space algorithm for \acyct for any $p\geq 1$ (\Cref{thm:acyc-t-ub}). This matches a $p$-pass $\Omega(n/p)$-space lower bound (up to logarithmic factors) for the problem, proven by \cite{ChakrabartiGMV20}. We also provide a simpler proof of the lower bound (\Cref{acyct-lbsimple}). Again, note that since our algorithm is deterministic and the lower bound holds for randomized algorithms, we fully settle the streaming complexity of the problem, even for multiple passes.

Finally, consider the sink (or equivalently, source) finding problem in DAGs. 
\begin{description}
    \item \sink: Given a DAG $G=(V,E)$, find a sink node, i.e., a node with outdegree 0, in $G$.
\end{description}
\noindent
We fully resolve the complexity of this problem by proving an $\Omega(n/p)$-space lower bound for any randomized $p$-pass algorithm (\Cref{thm:sinkfind-lb}). This means that the trivial $p$-pass $O(n/p)$-space algorithm which, in the $i$th pass, considers the $i$th set of $n/p$ nodes and checks whether any of them has outdegree zero, is optimal. 
    Note that Henzinger et al. \cite{HRR98} gave an $\Omega(n/p)$ space lower bound for sink finding in general digraphs that may not be DAGs, for which it is much easier to prove a lower bound since a sink may not exist. The lower bound is significantly more challenging to prove when the input is promised to be a DAG (ensuring the existence of a sink node).
    Further, our result establishes a gap between the complexities of \sink under general digraphs and tournaments: for \sinkt, \cite{ChakrabartiGMV20} showed that $\tO(n^{\Theta(1/p)})$-space is sufficient for $p$ passes. 

\mypar{Communication Complexity and Combinatorial Graph Theory} As a byproduct of our results, we resolve the communication complexity of multiple tournament problems (for their standard two-party communication versions) including \reacht, \sconnt, and \sccdect (see \Cref{sec:impli} for details), which might be of independent interest. Very recently, Mande et al.~\cite{mande2024communication} studied the communication complexity of several tournament problems; our results contribute to this study.

As further byproduct of our results, we make interesting observations on graph theoretic properties of tournaments that might be significant on their own. For instance, our results show that the indegree sequence of a tournament completely determines its SCCs. Again, for any two nodes $u$ and $v$, indegree of $u$ being smaller than that of $v$ implies a directed path from $u$ to $v$. See \Cref{sec:impli} for a compilation of such facts that, to the best of our knowledge, have not been documented in the literature. They should find applications in future combinatorial analysis or algorithmic design for tournaments.

\subsection{Other Related Work}\label{sec:related}

  Above we discussed prior works that are most relevant to ours. Here, we give an account of other related results.

  Coppersmith et al.~\cite{CoppersmithFR10} showed that simply sorting by indegree achieves a $5$-approximation to \fast, and hence, this implies a deterministic single-pass semi-streaming algorithm for $5$-approximation. Chakrabarti et al.~\cite{ChakrabartiGMV20} gave a one-pass $\tO(n/\eps^2)$-space algorithm for $(1+\eps)$-approximate \fast in exponential time and a $3$-approximation in $p$ passes and $\tO(n^{1+1/p})$ space in polynomial-time. The latter was improved by \cite{BawejaJW22} to a $(1+\eps)$-approximation under the same space-pass tradeoff and polynomial time. \cite{BawejaJW22} also established a space-time tradeoff for single-pass \fast algorithms. Guruswami, Velingker, and Velusamy \cite{GuruswamiVV17} studied the dual of the \fas problem, namely the {\em maximum acyclic subgraph} (MAS) problem, in the streaming model and showed that an algorithm that obtains a  better-than-$(7/8)$ approximation for MAS-size requires $\Omega(\sqrt{n})$ space in a single pass. Assadi et al.~\cite{AssadiKSY20} extended the lower bound to multiple passes, proving that $p$-pass $(1-\eps)$-approximation algorithms for MAS-size requires $\Omega(n^{1-\eps^{c/p}})$ space for some constant $c$. However, we do not know of any prior work that considered MAS specifically for streaming tournaments.

  The \acyc problem was considered by \cite{BawejaJW22} who gave a tight single-pass $\Omega(m \log (n^2/m))$ space lower bound for the problem, where $m$ is the number of edges. They showed a similar lower bound for testing reachability from a single node to all other nodes in a general digraph. Note that the $\Omega(n^{2-o(1)})$-space lower bound for \reach that Chen at al.~\cite{ChenKPSSY21-lognpass} proved for $O(\sqrt{\log n})$ passes, and Assadi and Raz \cite{AssadiR20} previously proved for $2$ passes, also apply to \acyc, as well as to the \fas (any multiplicative approximation) and topological sorting problems. For topological sorting of random graphs drawn from certain natural distributions, \cite{ChakrabartiGMV20} gave efficient random-order as well as adversarial-order streaming algorithms.
  
  Chakrabarti et al.~\cite{ChakrabartiGMV20} studied \sinkt on random-order streams and established an exponential separation between such streams and adversarial-order streams: the former allows a $\text{polylog}(n)$-space algorithm in just a single pass, but the latter necessitates $\Omega(n^{1/p}/p^2)$ space for $p$ passes. In a recent independent and parallel work, Mande, Paraashar, Sanyal, and Saurabh~\cite{mande2024communication} studied sink finding\footnote{They term it as \emph{source} finding, which is equivalent to sink finding.} in (not necessarily acyclic) tournaments in the two-party communication model and proved a tight bound of $\tTheta(\log^2 n)$ (where $\tTheta(f)$ hides polylog$(f)$ factors) on its (unrestricted round) communication complexity. In contrast, our lower bound proof for \sink shows its communication complexity to be $\Omega(n)$. Thus, this demonstrates an exponential gap between the communication complexity of sink-finding under tournaments and general digraphs (even if they are promised to be DAGs).

  Elkin \cite{Elkin17} gave multipass streaming algorithms for the problem of computing exact shortest paths on general digraphs, which implies algorithms for \stdist.  The SCC decomposition problem was studied by Laura and Santaroni \cite{lauraS11} on arbitrary digraphs under W-streams (allowing streaming outputs), where they designed an algorithm using $O(n\log n)$ space and $O(n)$ passes in the worst case. Other notable works on streaming digraphs include algorithms for simulating random walks \cite{SarmaGP11, Jin19, ChenKPSSY21-randomwalk}.

     \section{Preliminaries}\label{sec:prelims}

 \mypar{Notation and Terminology} All graphs in this paper are simple and directed. We typically denote a general digraph by $G = (V, E)$, where $V$ represents the set of vertices and $E \subseteq V \times V$ represents the set of directed edges. We usually denote a tournament by $T=(V,E)$. Throughout the paper, $n$ denotes the number of vertices of the graph in context. The notation $\tO(f), \tOmega(f),$ and $\tTheta(f)$ hide factors polylogarithmic in $f$. The outdegree and indegree of a vertex $v$ are denoted by $\dout(v)$ and $\din(v)$ respectively. For a positive integer $N$, the set $[N]:=\{1,\ldots,N\}$, and for integers $A\leq B$, the set $[A,B]:=\{A,\ldots,B\}$. When we say $\{u,v\}$ is a \emph{non-edge} in a digraph $G=(V,E)$, we mean that vertices $u,v\in V$ do not share any directed edge, i.e., $(u,v)\not\in E$ and $(v,u)\not\in E$. Again, we call a pair of edges $\{(u,v),(v,u)\}$ as \emph{bidirected edges}.

    \mypar{The Graph Streaming Model} The input graph $G=(V,E)$ is presented as follows. The set of $n$ nodes $V$ is fixed in advance and known to the algorithm. The elements of $E$ are inserted sequentially in a stream. The input graph and the stream order are chosen adversarially (worst case). We are allowed to make one or a few passes over the stream to obtain a solution. The goal is to optimize the space usage and the number of passes.
    
\smallskip
\noindent
     We use the following standard tool from the streaming literature.
    %A tournament graph (or tournament) $G = (V, E)$ is a simple directed graph such that for all unordered pairs of vertices $u$ and $v$ where $u \neq v$, there is exactly one directed edge between $u$ and $v$ in $G$. A majority of results in this paper deal with tournaments.

\begin{fact}[Sparse recovery \cite{gilbertI10, DasV13}]\label{fact:sp-rec}
    There is a deterministic $O(k\cdot\text{polylog}(M,N))$-space algorithm that receives streaming updates to a vector $\bx\in [-M,M]^N$ and, at the end of the stream, recovers $\bx$ in polynomial time if $\bx$ has at most $k$ non-zero entries.
\end{fact}

    \mypar{The Communication Model} All our streaming lower bounds are proven via reductions from problems in communication complexity. In this paper, we use the two-player communication model of Yao \cite{yao79}. Here, players Alice and Bob receive $\bx\in \{0,1\}^N$ and $\by\in \{0,1\}^N$ respectively, and must send messages back and forth to compute $\cF(\bx,\by)$ for some relation $\cF$ defined on $\{0,1\}^N \times \{0,1\}^N$. We consider the randomized setting where the players can communicate based on private random strings. The \emph{communication cost} of a randomized protocol for $\cF$ is defined as the maximum number of bits exchanged by the players, over all possible inputs $(\bx,\by)$ and all possible random strings. The \emph{randomized communication complexity} of a relation $\cF$, denoted by $R(\cF)$, is defined as the minimum communication cost over all randomized protocols that, given \emph{any} input $(\bx,\by)$, compute $\cF(\bx,\by)$ correctly with probability at least $2/3$. Analogously, the \emph{one-way randomized communication complexity} of $\cF$, denoted by $R^{\to}(\cF)$, is the minimum communication cost over all one-way communication protocols where Alice sends a single message to Bob, following which he must report the output.

    We use classical communication problems such as Index (\idx), Disjointness (\disj), and Set Intersection (\setint) to prove our lower bounds. The problems are formally defined below.  
    \begin{description}
        \item $\idx_N$: Alice holds a vector $\bx \in \{0, 1\}^N$, and Bob holds an index $i \in [N]$. The goal is to find $\bx_i$, the $i$th bit of $\bx$.

        \item $\disj_N$: Alice holds $\bx \in \{0, 1\}^N$ and Bob holds $\by \in \{0, 1\}^N$. The goal is to decide whether $\bx$ and $\by$ are \emph{disjoint} as sets, i.e., output whether there exists an index $i \in [N]$ such that $\bx_i = \by_i = 1$. We use the ``promise version'' (also known as \emph{unique disjointness}) where we are promised that the sets are either disjoint or have a \emph{unique} intersecting element $i$.

        \item $\setint_N$: Alice holds $\bx \in \{0, 1\}^N$ and Bob holds $\by \in \{0, 1\}^N$, where $\bx$ and $\by$ intersect at exactly one index, i.e., there is a unique $i \in [N]$ such that $\bx_i = \by_i = 1$. The goal is to output $i$.
    \end{description}

   We exploit the following known bounds on the communication complexity of the above problems. 

    \begin{fact}[\cite{Ablayev96}] \label{fact:idx-lb}
        The one-way randomized communication complexity $R^{\to}(\idx_N) = \Omega(N)$ .
    \end{fact}

    \begin{fact}[\cite{Razborov92, brodyckwy14}]\label{fact:disj-lb}
        The randomized communication complexity $R(\disj_N) = \Omega(N) = R(\setint_N)$.
    \end{fact}

    We also use the communication complexity of a variant of the \idx problem called \emph{Sparse Index} (\spidx).

\begin{description}
    \item  $\spidx_{N,k}$: This is the $\idx_N$ problem with the promise that Alice's vector $\bx$ has at most $k$ 1's, i.e., $|\{i : \bx_i =1\}|\leq k$. 
\end{description}
   
\begin{fact}[\cite{ChakrabartiCGT14}]\label{fact:sparse-index}
    $R^\to (\spidx_{N,k}) = \Omega(k)$
\end{fact}

    For convenience, we define a new communication problem $\tedge$, which is basically a version of the Index problem. 

    \begin{description}
        \item $\tedge_n$: Alice holds a tournament $T = (V, E)$ with $|V|=n$, and Bob holds a pair of vertices $u, v \in V$. The goal is to find the orientation of the edge between $u$ and $v$ in $T$.
    \end{description}
    
    This problem is equivalent to $\idx_{N}$ for $N=\binom{n}{2}$, and hence we get the following proposition.

    \begin{restatable}{proposition}{tedgelb} \label{prop:tedge-lb}
       The one-way randomized complexity $R^{\to}(\tedge_n) = \Omega(n^2)$. 
    \end{restatable}

\begin{proof}
Given any communication protocol $\Pi$ for $\tedge_n$, we can use it to solve $\idx_N$, where $N = \binom{n}{2}$, as follows. Let $\bx \in \{0, 1\}^N$ be Alice's vector and let $i \in [N]$ be Bob's index in the $\idx_N$ problem. Alice constructs a tournament graph $T = ([n], E)$ based on $\bx$: she treats the vector $\bx$ as being indexed by pairs $(u, v)$, where $u, v \in [n]$ and $u<v$. She can do this using any canonical bijection from $[N]$ to $\binom{[n]}{2}$. Then, if $\bx_{(u, v)} = 0$, Alice adds an edge in $T$ from vertex $u$ to vertex $v$, and if $\bx_{(u, v)} = 1$, she adds an edge from $v$ to $u$. This completes the construction of $T$. Let $(u^*,v^*)\in \binom{[n]}{2}$ be the pair corresponding to the index $i\in [N]$ that Bob holds. Alice and Bob can now use the protocol $\Pi$ to identify the orientation of the edge between the pair of nodes $u^*$ and $v^*$. If the edge is oriented from $u^*$ to $v^*$, Bob outputs that $\bx_i = 0$, otherwise he announces that $\bx_i = 1$. The correctness is immediate from construction. By \Cref{fact:idx-lb}, since  $R^{\to}(\idx_N)=\Omega(N)=\Omega(n^2)$, we get that the communication cost of $\Pi$ must be $\Omega(n^2)$. Therefore, $R^{\to}(\tedge_n) = \Omega(n^2)$.
    \end{proof}

    \section{Streaming Algorithms for SCC Decomposition and Applications}

In this section, we first design an algorithm for \sccdect in \Cref{sec:scctou}. Then we generalize it to arbitrary digraphs and obtain optimal algorithms for ``almost'' tournaments in \Cref{sec:sccarb}. 
Next, in \Cref{subsec:strconn-reach}, we note the immediate algorithms implied for \sconn and \reach in tournaments and almost tournaments, and then present lower bounds establishing their optimality. Additionally, \Cref{sec:otherapps} gives an account of other applications of our algorithms. 
In \Cref{subsec:hamcyc},  we apply our \sccdect framework to design the first streaming algorithm for \hcyct. Finally, in \Cref{subsec:smallscc}, we demonstrate how our \sccdect algorithm can be used to design efficient algorithms for \hpatht and \fast.

\subsection{Finding the SCC-graph}

First we design an algorithm for digraphs with zero non-edges and then show how to extend it to general digraphs.
  \subsubsection{Tournaments and Other Digraphs with No Non-Edges}\label{sec:scctou}
  
    We present a single-pass semi-streaming algorithm for \sccdect. In fact, our algorithm works for a broader class of digraphs: all digraphs without non-edges (i.e., every pair of vertices either shares exactly one directed edge or bidirected edges). Further, our algorithm solves the more general problem of finding the SCC-graph. Let $G_\scc$ denote the SCC-graph of a digraph $G$. It is easy to see that the SCC-graph of any digraph is always a DAG. Hence, for a digraph $G$ with no non-edge, its SCC-graph $G_\scc$ is an acyclic \emph{tournament}: between any two SCCs, there must exist an edge, but edges cannot exist in both directions. Since any acyclic tournament has a unique topological ordering (obtained by sorting its nodes by indegree), so does $G_\scc$.  Thus, $G_\scc$ can be simply and succinctly represented by a chain or permutation of $G$'s SCCs. The edges of $G_\scc$ are implicit: they go from left to right. Henceforth, we identify $G_\scc$ with its topological ordering, and our algorithm outputs this permutation.

    Formally, we prove the following theorem. 

    \begin{restatable}{theorem}{sccmain}\label{thm:scc-main}
    Given an input digraph $G=(V,E)$ with no non-edge, there is a deterministic single-pass algorithm that uses $O(n \log n)$ bits of space and
    outputs a partition $\langle V_1,\ldots,V_\ell\rangle$ of $V$ such that each $V_i$ is an SCC of $G$, and for every pair $i,j\in [\ell]$ with $i<j$, all edges between $V_i$ and $V_j$ are directed from $V_i$ to $V_j$.    
    \end{restatable}

     \begin{algorithm}[!ht]
  \caption{A single-pass semi-streaming algorithm for finding SCC-graph of a digraph with no non-edge
    \label{alg:sccgraph}}
  \begin{algorithmic}[1]
  \Statex \textbf{Input}: Stream of edge insertions of an $n$-vertex digraph $G=(V,E)$ that has no non-edge
  \Statex
    \Statex \underline{\textbf{Initialize}:}
    \State $d^-(v) \gets 0$ for each $v\in V$
    \Comment{In-degree counters that count $\din(v)$}
    \State $d^+(v) \gets 0$ for each $v\in V$
    \Comment{Out-degree counters that count  $\dout(v)$; {\color{teal}{not required if $G$ is a tournament}}}
    
  \Statex 
  \Statex\underline{\textbf{Process} (edge $(u,v)$):}
  \State $d^-(v) \gets d^-(v)+1$ \Comment{Increase in-degree of $v$}
  \State $d^+(u) \gets d^+(u)+1$ \Comment{Increase out-degree of $u$; {\color{teal}{not required if $G$ is a tournament}}}
  
\Statex 
      \Statex \underline{\textbf{Post-processing}:}
      \State $S \gets \emptyset$ \Comment{The current SCC}
      \State $G_\scc \gets \emptyset$ \Comment{The SCC-graph so far, stored as a list of SCCs in topological order}
      \State $n' \gets n$ \Comment{The size of the ``remaining graph'' $G\setminus G_\scc$}
      \State $c \gets 0$ \Comment{Initialize counter}
      \State $\langle v_1, \ldots, v_n \rangle \gets$ vertices sorted such that $d^-(v_1)\leq \ldots \leq d^-(v_n)$ \Comment{Ties broken arbitrarily}
      \For{$i$ = 1 to $n$}:
      \State Append $v_i$ to $S$ \Comment{Add vertex to current SCC}
      \State $c \gets c+ d^+(v_i)-d^-(v_i)$ \Comment{Update counter; {\color{teal}{if $G$ is a tournament, $d^+(v_i)=n'-1-d^-(v_i)$}}}
      \If{$c = |S| \cdot (n'-|S|)$} \Comment{Check if SCC is complete}
      \State Append $S$ to $G_\scc$ \Comment{Add SCC to the SCC-graph}
      \State $n' \gets n'-|S|$ \Comment{Update size of remaining graph}
      \For{$j$ = $i+1$ to $n$} \Comment{Update in-degrees in remaining graph}
      \State $d^-(v_j) \gets d^-(v_j) - |S|$
      \EndFor
      \State $S \gets \emptyset$; $c \gets 0$ \Comment{Empty $S$ and reset counter}
      \EndIf
      \EndFor
      \State Output $G_\scc$
      
  \end{algorithmic}
\end{algorithm}

The algorithm is given in \Cref{alg:sccgraph}. We prove its correctness using a series of lemmas and then prove the above theorem. For convenience, we define ``SCC-cuts'' that play an important role in the proof.

    \begin{definition}[SCC-cut]\label{def:scc-cut}
        Given a digraph $G=(V,E)$ with no non-edge,  for $U\subseteq V$, we call the cut ($U$, $V \setminus U$) an \emph{SCC-cut} if $U$ is the union of a prefix of SCCs in $G_\scc$. 
    \end{definition}

The next three lemmas prove important properties of SCC-cuts.

    \begin{restatable}{lemma}{scccutorient} \label{lem:scc-cut-orient}
        Given a digraph $G=(V,E)$ without any non-edge, a cut ($U$, $V \setminus U$) is an SCC-cut if and only if every edge between $U$ and $V \setminus U$ is oriented from $U$ to $V \setminus U$ in $G$. 
    \end{restatable}

    \begin{proof}%[Proof of \Cref{lem:scc-cut-orient}]
        Let $G_\scc$ be the ordering of $G$'s SCCs as described above.
        First suppose that ($U$, $V \setminus U$) is 
        an SCC-cut of $G$. Assume for the sake of 
        contradiction that there is an edge $e$ from a node 
        $v\in V \setminus U$ to a node $u\in U$. By \Cref{def:scc-cut}, all nodes in $U$ form a prefix of 
        $G_\scc$; this means the SCC containing $u$ appears in 
        $G_\scc$ before the SCC containing $v$. 
        Hence, $e$ must be a back-edge induced by the ordering $G_\scc$, 
        contradicting the fact 
        that $G_\scc$ is a topological ordering. 
        Thus, every edge between $U$ and $V \setminus 
        U$ in $G$ must be from $U$ to $V \setminus U$.
        
        Next, suppose that every edge between an 
        arbitrary vertex subset $U$ and its 
        complement $V \setminus U$ is oriented from 
        $U$ to $V \setminus U$ in $G$. This implies 
        that there is no path from a vertex in 
        $V\setminus U$ to a vertex in $U$. Therefore, 
        for each SCC $S$ of $G$, it must be that either $S\subseteq U$ or $S\subseteq V \setminus U$. Now, assume 
        for the sake of contradiction that there 
        exist SCCs $S$ and $S'$ of $G$ such that 
        $S\subseteq U$, $S'\subseteq V\setminus U$, and $S'$ 
        appears before $S$ in $G_\scc$. Since 
        $G_\scc$ is a topological ordering of an acyclic
        tournament, there must be an edge from $S'$ 
        to $S$, and hence, from $V \setminus U$ to 
        $U$, leading to a contradiction. Thus, all 
        SCCs of $G$ in $U$ must appear in $G_\scc$ 
        before all SCCs of $G$ in $V\setminus U$. 
        Hence, by \Cref{def:scc-cut}, ($U$, $V \setminus U$) is an SCC-cut of $G$.  
    \end{proof}

    \begin{restatable}{lemma}{scccuteqn} \label{lem:scc-cut-eqn}
        In a digraph $G=(V,E)$ with no non-edge, a cut $(U, V \setminus U)$ is an SCC-cut if and only if 

        \begin{equation}\label{eq:cutedges}
        \sum_{u \in U} (\dout(u)-\din(u)) = |U| \cdot |V \setminus U|    
        \end{equation}
         
    \end{restatable}

      \begin{proof}%[Proof of \Cref{lem:scc-cut-eqn}]
        Consider an SCC-cut $(U, V \setminus U)$ of $G$. By 
        \Cref{lem:scc-cut-orient}, we know that every 
        edge between $U$ and $V \setminus U$ is 
        oriented from $U$ to $V \setminus U$. Since 
        $G$ has no non-edge, there are exactly $|U| \cdot 
        |V \setminus U|$ such edges. Thus, the number of edges that 
        \emph{exit} $U$ (i.e., go from a vertex in $U$ to one outside $U$) is $|U| \cdot 
        |V \setminus U|$. Again, since there are no 
        edges oriented from $V \setminus U$ to $U$, 
        every edge entering a vertex in $U$ must be directed from another vertex in $U$. Thus, 
        $\sum_{u \in U} \dout(u)$ edges originate from some 
        vertex in $U$, and out of them, precisely 
        $\sum_{u \in U} \din(u)$ edges enter some 
        vertex in $U$. Therefore, the number of edges 
        that exit $U$ is exactly $\sum_{u \in U} 
        (\dout(u)-\din(u))$. Hence, we conclude
       that \cref{eq:cutedges} holds.
    
         Next, consider a cut $(U, V \setminus U)$ that satisfies \cref{eq:cutedges}. Note that 
         for \textit{any} $X\subseteq V$, we have 
         $$ \sum_{x \in X} (\dout(x)-\din(x)) = e(X,V\setminus X) - e(V\setminus X, X),$$
         where for any $A,B\subseteq V$, the function $e(A,B)$ is defined as the number of edges $(u,v)$ with $u\in A$ and $v\in B$.
         Therefore, by 
         \cref{eq:cutedges}, we have $|U| \cdot |V 
         \setminus U| = e(U,V\setminus U) - e(V\setminus U, U)$. Since $e(U,V\setminus U)\leq |U| \cdot |V 
         \setminus U|$, the equality only holds when $e(U,V\setminus U)= |U| \cdot |V 
         \setminus U|$ and $e(V\setminus U, U) = 0$. This precisely means that every edge between $U$ and $V \setminus U$ is oriented from $U$ to $V \setminus U$. By \Cref{lem:scc-cut-orient}, we know that $(U, V \setminus U)$ is an SCC-cut.
    \end{proof}

    \begin{restatable}{lemma}{sccincdeg} \label{lem:scc-inc-deg}
        Given a digraph $G=(V,E)$ without any non-edge and an SCC-cut $(U, V \setminus U)$ of $G$, we have $$\forall u \in U,~v \in V \setminus U:~~\din(u) < \din(v)$$ 
    \end{restatable}

    \begin{proof}%[Proof of \Cref{lem:scc-inc-deg}]
        Since $(U, V \setminus U)$ is an SCC-cut, every edge between $U$ and $V \setminus U$ is oriented from $U$ to $V \setminus U$ (\Cref{lem:scc-cut-orient}). Thus, since $G$ has no non-edge, for all $v \in V\setminus U$, we have $\din(v) \geq |U|$. Again, since a vertex in $U$ can receive edges only from other vertices in $U$, we have that $\din(u) < |U|$ for all $u \in U$.  Therefore, we conclude $$\forall u \in U,~v \in V \setminus U:~~\din(u) < \din(v)\qedhere$$ 
    \end{proof}

We are now ready to prove \Cref{thm:scc-main}.

    \begin{proof}[Proof of \Cref{thm:scc-main}] First, note that \Cref{alg:sccgraph} stores only the indegree and the outdegree of each node, and hence the space usage is clearly $O(n\log n)$ bits. Thus, it only remains to prove the correctness, i.e., the list of components it returns is indeed the SCC-graph $G_\scc$ of $G$. 
    
    Observe that the count $c$ computed by the algorithm is precisely $\sum_{u\in S}(\dout(u)-\din(u))$, for the set that the current state of $S$ contains. This is because whenever a new vertex $u$ is added to $S$, we add $\dout(u)-\din(u)$ to $c$ (see lines 9-11). Further, whenever the if condition on line 12 is satisfied, we empty $S$ and set $c$ to 0 at the end (line 17). We can thus use the count $c$ to check whether $S$ is an SCC of $G$ (by \Cref{lem:scc-cut-eqn}).
        
        We induct on the number of SCCs in $G$. Consider the base case when $G$ has a single SCC. By \Cref{def:scc-cut}, $G$ has only one SCC-cut $(V,\emptyset)$. 
        At any iteration $i<n$, the set $S$ contains a proper subset of $V$, and hence $(S, V\setminus S)$ cannot be an SCC-cut. Then, by \Cref{lem:scc-cut-eqn}, $c = \sum_{u \in S} (\dout(u)-\din(u))$ cannot equal $|S|\cdot (n-|S|)$. Since $n'=n$, the if condition on line 12 will not be satisfied. Therefore, at iteration $n$, the set $S$ contains all nodes, i.e., $S=V$. This means $(S, V\setminus S)$ is an SCC-cut, and by \Cref{lem:scc-cut-eqn}, $c$ must equal $|S|\cdot (n'-|S|)$ now. Thus, the if condition on line 12 will be satisfied, and $S=V$ will be added to $G_\scc$, after which the algorithm terminates. Hence, it will be the only set in the returned $G_\scc$, which is correct. This proves the base case.

         Next, assume by induction hypothesis that \Cref{alg:sccgraph} works correctly if the input tournament has $\ell$ SCC's for some $\ell\geq 1$. Now consider a tournament $G$ with $\ell+1$ SCCs, where $\langle V_1, \ldots, V_{\ell+1} \rangle$ represents $G_\scc$. Let $L$ be the list obtained by sorting the nodes of $G$ by in-degree (as in line 7). By \Cref{lem:scc-inc-deg}, we know that for all $u \in V_1$ and $v \in V \setminus V_1$, $\din(u) < \din(v)$. Hence, the first $|V_1|$ vertices in $L$ constitute $V_1$. Thus, when $S$ contains a prefix of $L$ of size $<|V_1|$, it basically contains a proper subset of $V_1$, meaning that the cut $(S,V\setminus S)$ cannot be an SCC-cut. Therefore, the if condition on line 12 will not be satisfied. Thus, $S$ will keep on adding nodes in the order given by $L$ until it contains $V_1$. Now, $(S, V \setminus S)$ is an SCC-cut. By \Cref{lem:scc-cut-eqn}, $c$ now equals $|S| \cdot (n'-|S|)$ (where $n'$ equals $n$, the size of $G$, because it wasn't updated since it was initialized to $n$ in line 5), and so we add $S = V_1$ to $G_\scc$, the list of SCCs (line 13). 
         
         We now show that lines 14-17 sets all the variables to their correct initial values for the tournament $G\setminus V_1$. Line 14 decrements $n'$ by $|S|$, and this is the size of $G\setminus V_1$ since $|S|=|V_1|$ and $n'$ was initially $n=|G|$. In line 15-16, we decrement the indegree $d^-(v)$ by $|S|$ for each $v$ in $G\setminus V_1$. This is the correct indegree of each such $v$ since  \Cref{lem:scc-cut-orient} implies that $v$ had an incoming edge from each vertex in $V_1$. Thus, when $V_1$ is removed, each $d^-(v)$ drops by exactly $|V_1|=|S|$. Finally, we empty $S$ and reset the counter $c$ to $0$. Note that the outdegrees $d^+(v)$ stay the same for each node $v$ in the remaining graph; this is because they had no outneighbor in $S$. Thus, as we continue to the next iteration, we have the variables $n'$, $c$, $S$, $d^-(v)$, and $d^-(v)$ for each $v\in T\setminus V_1$ set to their correct values for the tournament $G\setminus V_1$. Note that the list $L$ remains in the same order since each remaining vertex has its indegree dropped by the same value. The algorithm will now run on a graph with $\ell$ SCCs, with knowledge of its indegree and outdegree sequence. By the induction hypothesis, it will correctly find $\langle V_2, \ldots, V_{\ell+1} \rangle$ as the SCC graph of $G\setminus V_1$, which it will append after $\langle V_1\rangle$. Thus, our output is $\langle V_1, \ldots, V_{\ell+1} \rangle$, which is indeed the SCC-graph of $G$. By induction, we conclude that given any tournament $G$, our algorithm correctly finds the SCC-graph $G_\scc$.
\end{proof}

\begin{remark}
Since tournaments form a subclass of digraphs with no non-edges, \Cref{thm:scc-main} applies to them in particular. Also note that for tournaments, \Cref{alg:sccgraph} does not even need to store both indegrees and outdegrees; just one of them, say the set of indegrees, is enough since we can derive the outdegree values from them.   
\end{remark} 

\begin{corollary}\label{cor:sccdect}
     There is a deterministic single-pass $O(n\log n)$-space algorithm for \sccdect.
\end{corollary}

Observe that this space bound is optimal since $\Omega(n\log n)$ bits are needed to simply present the output of \sccdect.

\subsubsection{Generalization to Arbitrary Digraphs}\label{sec:sccarb}

We extend our algorithm to get sublinear-space solutions for a broader class of graphs. These are \emph{almost tournaments} which are $o(n^2)$ edges ``away'' from being a tournament. First, we formally define ``$k$-closeness to tournaments''. 
\begin{definition}[$k$-close to tournament]
    A digraph $G=(V,E)$ is called \emph{$k$-close to tournament} if a total of at most $k$ edges can be added to or deleted from $G$ such that the resulting digraph is a tournament. 
\end{definition}

In other words, a graph is $k$-close to tournament iff the total number of non-edges and bidirected edges it contains is at most~$k$. We obtain the following result for such graphs.

\begin{theorem} \label{thm:scc-gen}
    Given an input $n$-node digraph $G = (V, E)$ that is $k$-close to tournament, there is a deterministic single-pass streaming algorithm that finds the SCC-graph of $G$ in $\tO(n+k)$ space. 
\end{theorem}

To prove the theorem, we first need the following lemmas. For a digraph $G$, a \emph{source SCC} of $G$ is a source node in $G_\scc$. 

\begin{restatable}{lemma}{sourcescca}\label{lem:source-scc}
     Let $G=(V,E)$ be a digraph and $F$ be the set of its non-edges. Consider any non-empty $S\subseteq V$ that does
     not contain a source SCC of $G$ as a proper subset. For 
     each non-edge $f =\{u,v\} \in F$, we add the edge $(u,v)$ 
     to $G$ if $u\in S$ and $v\not\in S$, and otherwise, we 
     add either $(u, v)$ or $(v, u)$ (arbitrarily). Let $G'$ 
     be the resultant graph. Then $S$ is a source SCC in $G$ 
     if and only if it is a source SCC in $G'$. 
\end{restatable}

\begin{proof}%[Proof of \Cref{lem:source-scc}]
    For the ``only if'' direction, suppose that $S$ is a source SCC of $G$. First note that adding edges to a digraph cannot lead to removal of nodes from an SCC, and hence, nodes in $S$ stay in the same SCC in $G'$. Next, note that each edge in $G'$, which was added to $G$ between $S$ and $V \setminus S$, is oriented from $S$ to $V \setminus S$. Since all other edges in $G'$ between $S$ and $V\setminus S$ were already oriented outwards from $S$ in $G$ (since $S$ is a \emph{source} SCC of $G$), no new vertex can join $S$ to form a bigger SCC, and $S$ must still be a source SCC of $G'$. 

    For the ``if'' direction, we prove the contrapositive. Suppose that $S$ is not a source SCC of $G$. Then, we show that $S$ cannot be a source SCC of $G'$ either. We divide the analysis into 3 cases.

\mypar{Case 1} \textit{$S$ is not strongly connected in $G$}. Assume to the contrary that the addition of edges to $G$ turns $S$ into a source SCC in $G'$. Let $S'$ be any source SCC of $G$. Since $S$ does not contain $S'$ as a subset (by the premise of the lemma), there must be a vertex $u\in S'$ such that $u\not\in S$. Now assume that $S$ contains a vertex $v\in S'$. Since $u$ and $v$ are strongly connected in $G$, they are strongly connected in $G'$. Thus, since $S$ is an SCC in $G'$ and it contains $v$, it must also contain $u$: this is a contradiction. Hence, $S$ must be completely disjoint from $S'$, i.e., it must be disjoint from all source SCC's of $G$. This means there must be an edge from $V\setminus S$ to $S$ in $G$, and hence, in $G'$. This is a contradiction to the assumption that $S$ is a source SCC in $G'$. Therefore, $S$ cannot be a source SCC in $G'$.   

\mypar{Case 2} \textit{$S$ is a strict subset of an SCC of $G$}. Let $S'$ be the SCC of $G$ that contains $S$. Again, since addition of edges cannot remove nodes from an SCC, $S'$ remains strongly connected in $G'$. Hence $S$ cannot be a maximal connected component in $G'$, i.e., it cannot be an SCC in $G'$, let alone a source SCC.

    \mypar{Case 3} \textit{$S$ is a non-source SCC of $G$}. There must be some SCC $S'$ in $G$ such that there is an edge from $S'$ to $S$. Thus, there is an edge from $V \setminus S$ to $S$ in $G$, and hence, in $G'$. Therefore, $S$ cannot be a source SCC in $G'$. 
\end{proof}

\begin{restatable}{lemma}{sourcesccb}\label{lem:source-scc-2}
     Let $G=(V,E)$ be a digraph and $F$ be the set of its non-edges. Suppose we are only given the set $F$ and the indegrees and outdegrees of all nodes in $G$. Then, for any non-empty subset of nodes $S\subseteq V$ that does
     not contain a source SCC of $G$ as a proper subset, we can deterministically check whether $S$ is a source SCC of $G$.  
\end{restatable}

\begin{proof}%[Proof of \Cref{lem:source-scc-2}]
   Given the sets $S$ and $F$, we construct the graph $G'$ as described in \Cref{lem:source-scc} and calculate the resulting indegrees and outdegrees in $G'$. Observe that $G'$ has no non-edges. We thus use \Cref{alg:sccgraph} to find the SCC-graph of $G'$ and return that $S$ is a source SCC of $G$ if and only if it is the source SCC of $G'$. By \Cref{lem:source-scc}, this algorithm is correct.
\end{proof}

We are now ready to prove \Cref{thm:scc-gen}.

\begin{proof}[Proof of \Cref{thm:scc-gen}]
    Consider the following algorithm.

    \begin{description}
        \item[Before the stream:] Initialize the in-degree and out-degree of every vertex to 0. Initialize a sparse recovery sketch $\cS$ (\Cref{fact:sp-rec}) for a vector $\bx\in \{-1,0,1\}^{\binom{n}{2}}$ indexed by unordered pairs of vertices $\{u, v\}$. Initially $\bx_{\{i,j\}} = 1$ for all $\{i,j\}\in \binom{[n]}{2}$.
        
        \item[During the stream:] For edge $(u, v)$, increment the out-degree of $u$ and the in-degree of $v$, and update $\mathcal{S}$ by decrementing $\bx_{\{u, v\}}$ by 1.

        \item[After the stream:] Use $\mathcal{S}$ to recover the non-zero entries of $\bx$. For every unordered pair of vertices $\{u, v\}$, if $\bx_{\{u, v\}} = 1$ then add $\{u, v\}$ to the set of non-edges $F$. Iterate through subsets $S \subseteq V$ by increasing size, and check whether $S$ is a source SCC of $G$. If it is, we add $S$ to our list of SCC's and recurse on $G\setminus S$ with the indegrees updated as follows. For each $v \in V\setminus S$, let $\din^{(S)}(v):=|\{u\in S: \{u,v\}\not\in F\}|$. Then we decrement the current value of $\din(v)$ by $\din^{(S)}(v)$. The outdegrees of nodes in $V\setminus S$ are not updated. 

        The algorithm terminates when the remaining graph becomes empty. Suppose that the list of SCC's we obtain is $\langle S_1, ..., S_\ell \rangle$ (in this order). Construct $G_\scc$ as follows: for all $i,j\in [\ell]$ with $i<j$, add an edge from $S_i$ to $S_j$ if there exists $u \in S_i$ and $v \in S_j$ such that $\{u,v\}\not\in F$. Return the resultant SCC-graph as $G_\scc$.
    \end{description}

    Note that the space bound of $\tO(n + k)$ space follows since counting the degrees takes $\tO(n)$ space and the sketch $\cS$ takes $O(k\cdot \text{polylog} n)$ space (\Cref{fact:sp-rec}). It remains to prove the correctness of the algorithm.
    Since we initialize the vector $\bx$ to the all $1$s vector and decrement $\bx_{\{u,v\}}$ whenever a directed edge between the vertices $u$ and $v$ appears, we get that at the end of the stream, the number of edges between a pair of vertices $\{u, v\}$ is equal to $1-\bx_{\{u, v\}}$. Hence, the entries where $\bx_{\{u, v\}} = 1$ are precisely the non-edges. Further, the non-zero entries of $\bx$ correspond to either the non-edges (when $\bx_{\{u, v\}}=1$) or the bidirected edges (when $\bx_{\{u, v\}}=-1$), and since $G$ is $k$-close to being a tournament, $\bx$ has at most $k$ non-zero entries. Hence, $\cS$ works correctly and we correctly identify $F$.
    
    Next, note that we iterate over subsets $S$ by increasing size to check whether each is a source SCC of $G$, and whenever the answer is positive, we remove $S$ and recurse on the remaining graph. Thus, no $S$ we check can contain a source SCC of $G$ as a strict subset. This implies that we can correctly apply \Cref{lem:source-scc-2} and identify the source SCC in the remaining graph. 

    To see that we recurse on $G\setminus S$ with its accurate set of indegrees and outdegrees, observe that $\din^{(S)}(v)$ counts the number of edges directed from some node in $S$ to $v$: for each $u\in S$ with $\{u,v\}\not\in F$, we know that there is an edge between $u$ and $v$, and since $S$ is a source SCC of $G$, it must be oriented from $u$ to $v$. Hence, removal of $S$ decreases the indegree of $v$ by precisely $\din^{(S)}(v)$. Again, since there was no edge from any $v\in V\setminus S$ to $S$, there is no change in the outdegrees.   
    
    Finally, since we obtain the list $\langle S_1, ..., S_\ell \rangle$ by recursively removing source SCC's from the remaining graph and appending them to the list, we know that $\langle S_1, ..., S_\ell \rangle$ is a topological ordering of $G_\scc$. Thus, edges can only go from left to right in this ordering: for all $i < j$, there is an edge from $S_i$ to $S_j$ if and only if, for some $u \in S_i$ and some $v \in S_j$, the pair $\{u,v\}$ does not appear in the set of non-edges $F$. Therefore, SCC-graph returned by our algorithm is indeed $G_\scc$.
\end{proof}

In particular, we get the following corollary from \Cref{thm:scc-gen}.

\begin{corollary}\label{cor:scc-gen-ub}
    Given an input $n$-node digraph $G = (V, E)$ that is $k$-close to tournament, there is a deterministic $\tO(n+k)$-space streaming algorithm for \sccdec. 
\end{corollary}

A matching lower bound of $\Omega(n+k)$ space follows immediately from the lower bounds for \sconn or \reach (\Cref{thm:reach-gen-lb}) that we prove in \Cref{subsec:strconn-reach}. 

\begin{corollary}\label{cor:scc-gen-lb}
     A single-pass streaming algorithm that solves \sccdec on any digraph that is $k$-close to tournament requires $\Omega(n+k)$ space. 
\end{corollary}

 Note that the algorithm described in the proof of \Cref{thm:scc-gen} makes $\tO(2^n)$ calls to the protocol in \Cref{lem:source-scc-2}, each of which takes $\tO(n)$ time to check. Thus, the algorithm takes $\tO(2^n)$ time in the worst case. Even for the easier problem of \sconn, we show that such a runtime in necessary in some sense, as long as we use a subroutine for tournaments to solve the problem for almost-tournaments (see \Cref{thm:sconn-oracle-lb}).

\subsection{Reachability and Strong Connectivity}\label{subsec:strconn-reach}\label{sec:conn}

Here, we complement the upper bounds implied from the previous section for \reach and \sconn (see \Cref{sec:results} for formal definition) with matching lower bounds. First we prove the lower bounds for tournaments and then for arbitrary digraphs. Additionally, we show that the runtime given by our \sccdec algorithm for almost tournaments can be improved for the special cases of \reach and \sconn.

\subsubsection{Tight Lower Bounds for Tournaments}\label{subsec:tou-lb}

Observe that \Cref{thm:scc-main} immediately implies single-pass semi-streaming algorithms for \sconnt and \reacht. Given the SCC-graph $\langle V_1, \ldots, V_\ell \rangle$ of $T$, a node $t$ is reachable from another node $s$ in $T$ iff $i \leq j$, where $V_i$ is the SCC containing $s$ and $V_j$ is the SCC containing $t$. Again, a tournament $T$ is strongly connected iff the SCC-graph contains a single SCC. 
Thus, we get the following corollaries.

\begin{restatable}{corollary}{reachable}\label{cor:reacht}
    There is a deterministic single-pass $O(n\log n)$-space algorithm for \reacht.
\end{restatable}

\begin{restatable}{corollary}{strongconn}\label{cor:strongconnt}
    There is a deterministic single-pass $O(n\log n)$-space algorithm for \sconnt.
\end{restatable}

We now prove the space bound for these problems to be tight (up to polylogarithmic factors) even for polylog$(n)$ passes. Note that while our algorithms work for any digraph without non-edges, our lower bounds apply even for the easier versions where the inputs are promised to be tournaments.

\begin{restatable}{theorem}{reachlb}\label{thm:reach-lb}
     Any randomized $p$-pass algorithm that solves \reacht requires $\Omega(n/p)$ space.
\end{restatable} 

 \begin{proof}

    \begin{figure}[H]
        \centering
        
        \includegraphics[scale=0.7]{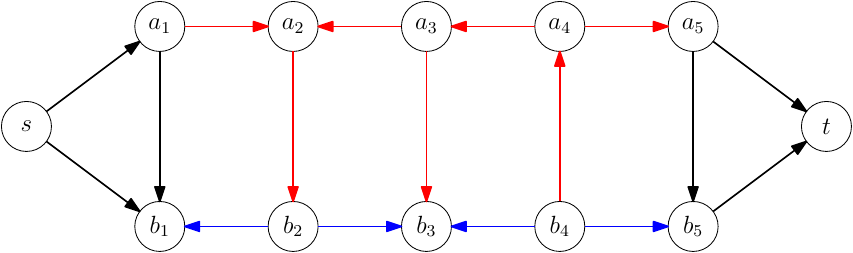}
        
        \caption{The graph constructed in the lower bound for \reacht for $N=4$ when Alice has the vector $\bx=0110$ and Bob has the vector $\by=1010$. Edges that Alice (resp. Bob) adds are colored red (resp. blue), and input-independent edges are colored black. Other edges independent of the input have been excluded for clarity, and always go from right to left.}
        
        \label{fig:stconn-t-lb}
    \end{figure}

        We reduce from $\disj_N$ (see \Cref{sec:prelims} for definition). Alice and Bob construct a tournament $T = (V, E)$ with $n = 2N+4$ vertices, where vertices are labeled $\{s, a_1, b_1, \ldots, a_{N+1}, b_{N+1}, t\}$ (see \Cref{fig:stconn-t-lb}). Alice, based on her input $\bx$, adds edge $(a_{i+1}, a_i)$ if $\bx_i = 1$, and edge $(a_i, a_{i+1})$ if $\bx_i = 0$. Similarly, Bob, based on his input $\by$, adds edge $(b_{i+1}, b_i)$ if $\by_i = 1$, and edge $(b_i, b_{i+1})$ if $\by_i = 0$. Additionally, for $2 \leq i \leq N$, Alice adds edge $(a_i, b_i)$ if $\bx_i = 1$, and edge $(b_i, a_i)$ if $\bx_i = 0$. Finally, independent of the inputs, one player (say Bob) adds edges $(s, a_1), (s, b_1), (a_1, b_1), (a_{N+1}, b_{N+1}), (a_{N+1}, t)$ and $(b_{N+1}, t)$. He also adds edges from $\{a_i, b_i\}$ to $s$ for $i \neq 1$, from $t$ to $\{a_i, b_i\}$ for $i \neq N+1$, from $t$ to $s$, and from $a_i$ to $b_j$ and $b_i$ to $a_j$ for all $j < i$. This completes the construction of $T$. It is easy to check that $T$ is a tournament.  

        Let us show that $T$ contains a path from $s$ to $t$ if and only if $\bx$ and $\by$ are disjoint as sets. First suppose that $\bx$ and $\by$ are disjoint. We will use induction to show that we can get to either $a_i$ or $b_i$ for all $i$. As the base case, note that we can get from $s$ to both $a_1$ and $b_1$. Assume that we can get from $s$ to either $a_i$ or $b_i$. Since $\bx$ and $\by$ are disjoint, we know that either $\bx_i = 0$ or $\by_i = 0$. Thus, there is either an edge from $a_i$ to $a_{i+1}$ or from $b_i$ to $b_{i+1}$. If we can get to $a_i$ or $b_i$ and have edge $(a_i, a_{i+1})$ or edge $(b_i, b_{i+1})$ respectively, our inductive hypothesis is satisfied. Next, consider the case where we can get to $a_i$ but do not have the edge $(a_i, a_{i+1})$. In this case, we know that $\bx_i = 1$, so Alice must have added the edge $(a_i, b_i)$. Thus, we can get to $b_i$, and there must be an edge from $b_i$ to $b_{i+1}$, satisfying the inductive hypothesis. This leaves the case where we can get to $b_i$ but do not have the edge $(b_i, b_{i+1})$. This implies that $\by_i = 1$, so $\bx_i = 0$. Thus, Alice must have added the edge $(b_i, a_i)$ and $(a_i, a_{i+1})$, satisfying the inductive hypothesis. By induction, we can conclude that there is a path from $s$ to either $a_{N+1}$ or $b_{N+1}$, and we know that there is an edge from both $a_{N+1}$ and $b_{N+1}$ to $t$, so there is a path from $s$ to $t$. 

        For the other direction, assume that $\bx$ and $\by$ are not disjoint, and let $i$ be the index at which they intersect. Consider the cut in $T$ that partitions the vertices into subsets $L = \{s, a_1, b_1, \ldots, a_i, b_i\}$ and $R = \{a_{i+1}, b_{i+1}, \ldots, a_{N+1}, b_{N+1}, t\}$. Note that there are no edges from $L$ to $R$; Alice adds the edge $(a_{i+1}, a_i)$, Bob adds the edge $(b_{i+1}, b_i)$, and every other edge in the cut is input-independent and goes from $R$ to $L$. Thus, there is no path from $s$ to $t$. 

        This implies that we can use a streaming algorithm for \reacht on $n = 2N+4$ vertices to solve $\disj_N$. Since $\disj_N$ requires $\Omega(N)$ total bits of communication (see \Cref{fact:disj-lb}), \reacht must require $\Omega(n/p)$ space in $p$ passes.
    \end{proof}

We obtain a similar lower bound for \sconnt.

 \begin{restatable}{theorem}{sclb}\label{thm:strongconn-lb}
     Any randomized $p$-pass algorithm that solves \sconnt requires $\Omega(n/p)$ space.
\end{restatable} 

\begin{proof}%[Proof of \Cref{thm:strongconn-lb}]
        The construction of the graph $G$ is very similar to that in the \reach lower bound (proof of \Cref{thm:reach-lb}), with one minor modification. If $\by_1 = 0$, which is when Bob adds edge $(b_1, b_2)$, he also adds edge $(a_1, b_1)$. On the other hand, if $\by_1 = 1$, which is when Bob adds edge $(b_2, b_1)$, he also adds edge $(b_1, a_1)$. Similarly, if $\by_N = 0$, which is when Bob adds edge $(b_N, b_{N+1})$, he also adds edge $(b_{N+1}, a_{N+1})$, and if $\by_N = 1$, which is when Bob adds edge $(b_{N+1}, b_N)$, he also adds edge $(a_{N+1}, b_{N+1})$. Note that these edges were fixed in the earlier construction. The reason for this modification is to avoid accidental sources or sinks, as will become more clear further into the proof. 
        
        Let us show that $G$ is strongly connected 
        if and only if $\bx$ and $\by$ are disjoint. 
        Assume that $\bx$ and $\by$ are disjoint. As 
        shown in the proof above, this implies that 
        there is a path from $s$ to $t$ in $G$. 
        Moreover, we know by our construction that 
        there is an edge from $t$ to $s$. Thus, $s$ 
        and $t$ are in a single SCC. Next, note that 
        for indices $2 \leq i \leq N$, both $a_i$ 
        and $b_i$ have an edge from $t$ and to $s$. 
        Thus, all such vertices are also in the same 
        SCC as $s$ and $t$. It remains to prove that $a_1,b_1,a_{N+1},b_{N+1}$ are in the same SCC.

        We have either $(b_1,b_2)\in G$ or $(b_2,b_1)\in G$. First consider the former case. Since the edges $(s,b_1), (b_1,b_2), (b_2,s) \in G$, we conclude that $b_1$ is in the same SCC as $s$ and $b_2$. Again, due to the modification in our 
        construction stated at the beginning of the 
        proof, $G$ has the edge $(a_1, 
        b_1)$. Since $(s,a_1), (a_1,b_1)\in G$ and $s$ and $b_1$ are in the same SCC, so is $a_1$. In the other case when $(b_2,b_1)\in G$, by construction we know that $\by_1= 1$. Since $\bx$ and $\by$ are disjoint, $\bx_1=0$ and hence $(a_1,a_2)\in G$. By similar logic as above, we first derive that $a_1$ is in the same SCC as $s$, and since $(b_1,a_1)\in G$ in this case (modification in this proof), we get that so is $b_1$.
        
        Finally, let us consider vertices $a_{N+1}$ and $b_{N+1}$, for which a symmetric logic holds. Both $a_{N+1}$ and $b_{N+1}$ have outgoing edges to $t$. $G$ contains either $(b_N, b_{N+1})$ or $(b_{N+1}, b_{N})$. In the former case, we get $b_{N+1}$ is a part of the same SCC, and since $G$ contains $(b_{N+1}, a_{N+1})$, so is $a_{N+1}$. In the other case when $(b_{N+1}, b_{N})$, we know $\by=1$ and so $\bx=0$. Hence, $G$ contains $(a_N, a_{N+1})$, which first tells us that $a_{N+1}$ is in the same SCC. Then, since $(a_{N+1}, b_{N+1})\in G$, so is  $b_{N+1}$. Thus, we see that every vertex of $G$ is in the same SCC, so $G$ is strongly connected.

        For the reverse direction, assume that $\bx$ and $\by$ are not disjoint. As shown in the proof of \Cref{thm:reach-lb}, this implies that there is no path from $s$ to $t$, so by definition $G$ is not strongly connected. 

        This implies that we can use a streaming algorithm for checking strong connectivity to solve $\disj_N$. Since $\disj_N$ requires $\Omega(N)$ total bits of communication (see \Cref{fact:disj-lb}), \sconn must require $\Omega(n/p)$ space in $p$ passes.
    \end{proof}

\subsubsection{Tight Lower Bounds for Arbitrary Digraphs}\label{subsec:arbdig-lb} 

Observe that \Cref{thm:scc-gen} immediately implies an $\tO(n+k)$-space algorithm for \sconn and \reach on digraphs that are $k$-close to tournaments. Here, we prove that the space bound is tight (up to polylogarithmic factors) for these problems. Hence, the same lower bound applies to \sccdec, as we noted in \Cref{cor:scc-gen-lb}. This demonstrates how the complexities of the \sccdec, \sconn, and \reach problems smoothly increase with the ``distance'' from the tournament property.

\begin{restatable}{theorem}{reachgenlb}\label{thm:reach-gen-lb}
     A single-pass streaming algorithm that solves \reach or \sconn on any digraph that is $k$-close to tournament requires $\Omega(n+k)$ space. 
\end{restatable}

\begin{proof}

    \begin{figure}[H]
        \centering
        
        \includegraphics[scale=0.8]{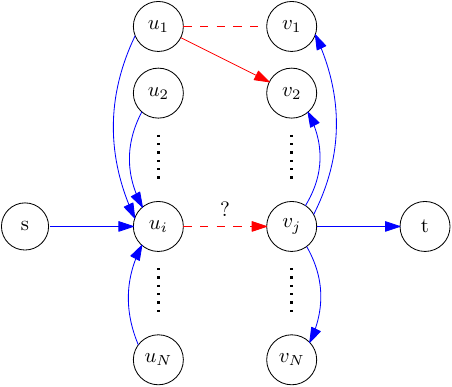}
        
        \caption{The graph constructed in the lower bound for \reach and \sconn where $x_{1, 1} = 1$ and $x_{1, 2} = 0$. Edges that Alice (resp. Bob) adds are colored red (resp. blue). The dotted lines represent the non-edges. All other edges have been excluded for clarity; edges incident on $s$ or $t$ go from right to left, and edges between $u_a$'s where $a \neq i$ are oriented arbitrarily, as are edges between $v_b$'s for $b \neq j$.}
        
        \label{fig:stconn-lb}
    \end{figure}

    We reduce from $\spidx_{N^2, k}$. The same construction will work for both \reach and \sconn. For the rest of the proof we use a canonical bijection from $[N^2]$ to $[N] \times [N]$, indexing by pairs $(i, j)$ where $i, j \in [N]$. Alice and Bob construct a graph $G = (V, E)$ on $n = 2N+2$ vertices (See \Cref{fig:stconn-lb}). Let us label these vertices $s, u_1, ..., u_{N}, v_1, ..., v_{N}, t$. Given a vector $\bx$, for each $a,b\in [N]$, Alice adds the edge $(u_a, v_b)$ if $\bx_{a, b} = 0$. Given an index $(i, j)$, Bob adds the following edges.
    \begin{itemize}
        \item $(s, u_i)$, $(v_j, t)$, $(v_j, s)$, $(t, u_i)$, and $(t, s)$.
        \item $(u_{i'}, s)$, $(t, u_{i'})$ and $(u_{i'}, u_i)$ for all $i' \neq i$ and $(v_{j'}, s)$, $(t, v_{j'})$ and $(v_j, v_{j'})$ for all $j' \neq j$.
        \item $(u_{i'}, u_{i''})$ for all $i' < i''$ where $i', i'' \neq i$ and $(v_{j'}, v_{j''})$ for all $j' < j''$ where $j', j'' \neq j$. 
    \end{itemize}
    This completes the construction of graph $G$, and it is easy to verify that the number of unordered pairs of vertices $u, v \in V$ such that $(u, v) \not \in E$ and $(v, u) \not \in E$ is precisely $k$, the number of 1s in $\bx$. If there is a path from $s$ to $t$ in $G$, Bob will announce that the value at the index is 0, and if there is not he will announce that the value at the index is 1. 

    Let us first show that $G$ contains a path from $s$ to $t$ if and only if $\bx_{i, j} = 0$. If $\bx_{i, j} = 0$, we know that there is a path from $s$ to $t$ using the edges $(s, u_i)$, $(u_i, v_j)$ and $(v_j, t)$. Next, consider the case where $\bx_{i, j} = 1$. By construction, we know that the only edge out of $s$ is to vertex $u_i$ and the only edge to vertex $t$ is from $v_j$. Thus, it suffices to show that there is no path from $u_i$ to $v_j$. By construction, we know that the only edges out of $u_i$ are of the form $(u_i, v_{j'})$ for some $j' \neq j$. However, we also know that the only edges out of all such $v_{j'}$ are either to $v_{j''}$, where $j'' \neq j$, or to $s$, which we have already considered. Thus, there is no path from $u_i$ to $v_j$, and therefore no path from $s$ to $t$. 

    This implies that we can use a streaming algorithm for \reach on a graph $G$ with $n = 2N+2$ vertices and $k$ non-edges to solve $\spidx_{N^2, k}$. Since $\spidx_{N^2, k}$ requires $\Omega(k)$ bits of one-way communication (see \Cref{fact:sparse-index}), \reach must require $\Omega(k)$ space in a single pass. For $k \leq n$, the construction in \Cref{thm:reach-lb} gives an $\Omega(n)$ lower bound since the graph can be modified to have $k$ non-edges by removing some of the input-independent edges that go from right to left, following which the reduction still works. Thus, \reach on digraphs $k$-close to tournaments requires $\Omega(n+k)$ space in a single pass.

    To get the same bound for \sconn, note that $G$ is strongly connected if and only if $\bx_{i, j} = 0$. If $\bx_{i, j} = 1$, there is no path from $s$ to $t$ (as shown above), meaning that $G$ is not strongly connected. Otherwise, if $\bx_{i, j} = 0$, then we have the cycle $(s,u_i,v_j,t,s)$. Hence, $s, u_i, v_j, t$ are in the same SCC. For any other vertex $w$ in $G$, observe that there is an edge from $t$ to $w$ and from $w$ to $s$. Hence, every such $w$ is in the same SCC as well, i.e., $G$ is strongly connected. Thus, similar to \reach, we apply \Cref{fact:sparse-index} to get a single-pass lower bound of $\Omega(k)$ and, for $k< n$, modify the graph in the proof of \Cref{thm:strongconn-lb}  to have $k$ non-edges and apply the same reduction arguments to obtain a lower bound of $\Omega(n)$. Hence, the lower bound for \sconn on digraphs $k$-close to tournaments is strengthened to $\Omega(n+k)$.
\end{proof}

\subsubsection{Improved and ``Optimal'' Time Complexity}\label{subsec:time} 

Recall that our \sccdec algorithm for digraphs $k$-close to tournaments (\Cref{thm:scc-gen}) takes $\tO(2^n)$ time. We show that this runtime can be improved for small $k$ when solving \sconn and \reach (rather than the general \sccdec problem): our modified algorithms use the same space but take $\tO(2^k\cdot n)$ time if $k<n$. Subsequently, we show that the runtime of this algorithm is optimal in some sense (see \Cref{thm:sconn-oracle-lb}). 

\mypar{Faster Algorithms} First we prove the following theorem.

\begin{restatable}{theorem}{reachgenub}\label{thm:reach-gen-ub}
    Given an input $n$-node digraph $G = (V, E)$ that is $k$-close to a tournament, there are deterministic single-pass streaming algorithms that solve \reach and \sconn on $G$ using $\tO(n+k)$ space and $\tO(2^{\min(k, n)})$ time for each problem. 
\end{restatable}

We define the \emph{completion} of a digraph that helps in formalizing the proof. Informally, a completion of a digraph $G$ is any digraph $G'$ that is obtained by filling each non-edge of $G$ with an arbitrary edge. Thus, $G'$ has no non-edge. In particular, if $G$ is a tournament, so is $G'$. 

\begin{definition}[completion of a digraph]
    Given any digraph $G=(V,E)$, a completion $G'=(V,E')$ of $G$ is a digraph such that $E'\supseteq E$ and for every non-edge $f=\{u,v\}\in G$, exactly one of the edges $(u,v)$ and $(v,u)$ is in $E'$.
\end{definition}

We begin by proving the following series of lemmas.

\begin{lemma}\label{lem:toporder}
    Given a DAG $G = (V, E)$ with $|V|=n$ and a pair of vertices $s, t \in V$, if $t$ is not reachable from $s$, then there exists a topological ordering $\sigma : V \to [n]$ such that $\sigma(t) < \sigma(s)$. 
\end{lemma}

\begin{proof}
    Assume that there is no path from $s$ to $t$. Let $\text{IN}(t) \subseteq V$ be the vertices from which $t$ is reachable, and let $\text{OUT}(s) \subseteq V$ be the vertices which are reachable from $s$. Note that $t \in \text{IN}(t)$ and $s \in \text{OUT}(s)$. Also, since $t$ is not reachable from $s$, $\text{IN}(t) \cap \text{OUT}(s) = \emptyset$. Let $\text{REM} = V \setminus (\text{IN}(t) \sqcup \text{OUT}(s))$.
    
    Let us show that $\langle \text{IN}(t), \text{REM}, \text{OUT}(s)\rangle$ is a valid topological ordering of $G$, where each subset is ordered as a topological ordering of its induced subgraph. There can not be an edge from a vertex $u \in \text{REM}$ to a vertex $v \in \text{IN}(t)$, since then $u$ would be in $\text{IN}(t)$. Similarly, there can not be an edge from a vertex $u \in \text{OUT}(s)$ to a vertex $v \in \text{REM}$, since then $v$ would be in $\text{OUT}(s)$. Finally, there can not be an edge from a vertex $u \in \text{OUT}(s)$ to a vertex $v \in \text{IN}(t)$, since then there would be a path from $s$ to $t$ via this edge. Thus, this is a valid topological ordering in which $t$ appears before $s$, so such a topological ordering must exist if $t$ is not reachable from $s$.
\end{proof}

\begin{lemma}\label{lem:streach}
    Given a digraph $G = (V, E)$ and a pair of vertices $s, t \in V$, the node $t$ is reachable from the node $s$ in $G$ if and only if $t$ is reachable from $s$ in all possible completions $G' = (V, E')$ of $G$. 
\end{lemma}

\begin{proof}
    First, note that if there is a path from $s$ to $t$ in $G$, then this path must be also present in all $G'$ since $E' \supseteq E$. Thus, let us show that the other direction also holds. Assume that $G$ does not contain a path from $s$ to $t$. For a vertex $v \in V$, let $\scc(v)$ be the SCC of $G$ containing $v$. Since $G$ does not contain a path from $s$ to $t$, we know that $G_{\scc}$ does not contain a path from $\scc(s)$ to $\scc(t)$. Thus, since $G_{\scc}$ is acyclic, by \Cref{lem:toporder} there is a topological ordering $\sigma$ of $G_{\scc}$ such that $\sigma(\scc(t)) < \sigma(\scc(s))$. Consider the completion $G' = (V, E')$ obtained as follows. For a pair of vertices $\{u, v\} \in F$, where $F$ is the set of $G$'s non-edges, add $(u, v) \in E'$ if $\sigma(\scc(u)) \leq \sigma(\scc(v))$ and $(v, u) \in E'$ if $\sigma(\scc(v)) \leq \sigma(\scc(u))$. Note that none of the additional edges are back edges according to $\sigma$, implying that $\sigma$ is also a valid topological ordering of $G'_{\scc}$. Hence, we get a completion of $G'$ where there is no path from $s$ to $t$, completing the proof. 
\end{proof}

We are now ready to prove \Cref{thm:reach-gen-ub}.

\begin{proof}[Proof of \Cref{thm:reach-gen-ub}]
    First, note that if $k > n$, we can use the algorithm in \Cref{thm:scc-gen} to obtain the SCC-graph of $G$, allowing us to solve \reach on $G$ in $\tO(n + k)$ space and $\tO(2^n)$ time. For $k < n$, consider the following algorithm.
    
    \begin{description}
        \item[Before the stream:] Initialize the in-degree and out-degree of every vertex to 0. Initialize an instance of sparse recovery $\mathcal{S}$ (\Cref{fact:sp-rec}) on a vector $\bx\in \{-1,0,1\}^{\binom{n}{2}}$ indexed by unordered pairs of vertices $\{u, v\}$. Initially $\bx_{\{i,j\}} = 1$ for all $\{i,j\}\in \binom{[n]}{2}$.
        
        \item[During the stream:] For edge $(u, v)$, increment the out-degree of $u$ and the in-degree of $v$, and update $\mathcal{S}$ by decrementing $\bx_{\{u, v\}}$ by 1.

        \item[After the stream:] Use $\mathcal{S}$ to recover the non-zero entries of $\bx$. For every unordered pair of vertices $\{u, v\}$, if $\bx_{\{u, v\}} = 1$ then add $\{u, v\}$ to the set of non-edges $F$. Iterate over all possible completions $G'$ obtained by adding $(u, v)$ or $(v, u)$ for every unordered pair of vertices $\{u, v\} \in F$. Update the degrees of $G$ to obtain the degrees of $G'$. Use the algorithm from \Cref{cor:reacht} to check if $t$ is reachable from $s$ in $G'$. If there exists a $G'$ in which $s$ is not reachable from $t$, output false. If there is no such $G'$, output true.
    \end{description}

    Observe that the algorithm only takes $\tO(n + k)$ space since counting the degrees takes $\tO(n)$ space and the sketch $\cS$ takes $\tO(k)$ space (see \Cref{fact:sp-rec}). Similar to the proof of \Cref{thm:scc-gen}, we argue that the number of edges between a pair of vertices $\{u, v\}$ is equal to $1-\bx_{\{u, v\}}$. Hence, the pairs $\{u, v\}$ with $\bx_{\{u, v\}} = 1$ are the non-edges. Also, the non-zero entries of $\bx$ correspond to either the non-edges (when $\bx_{\{u, v\}}=1$) or the bidirected edges (when $\bx_{\{u, v\}}=-1$): this means $\bx$ has at most $k$ non-zero entries since $G$ is $k$-close to tournament. Hence, $\cS$ works correctly and we correctly identify $F$. Then, the correctness of the algorithm follows directly from \Cref{lem:streach} since we know that $t$ is reachable from $s$ in $G$ if and only if $t$ is reachable from $s$ in all completions $G'$.

    Finally, note that our algorithm only requires checking $O(2^k)$ possible completions since there are at most $k$ non-edges. For each completion, we take $\tO(n)$ time to check reachability (see \Cref{alg:sccgraph}). Thus, this algorithm takes $\tO(2^k)$ time, giving an overall runtime $\tO(2^{\min(n, k)})$ for any $k$. 

    The proof for \sconn is along very similar lines and we outline it. We run the same protocol before and during the stream, but during post-processing we return that $G$ is strongly connected if and only if every possible completion $G'$ of $G$ is strongly connected. The correctness follows along similar lines. If $G$ is strongly connected, then so is each of its completions $G'$. Again, if $G$ is not strongly connected, then there exist vertices $s$ and $t$ in $G$, such that there is no directed path from $s$ to $t$. Then, as in the proof of \Cref{lem:streach}, we can construct a completion $G'$ with no path from $s$ to $t$, meaning that $G'$ is not strongly connected. The runtime follows by similar arguments as we go over all possible $2^k$ completions.
\end{proof}

\mypar{A Query Lower Bound} To justify the exponential runtime, we show a query lower bound. Let us consider \sconn. Suppose, similar to us, an $\tO(n+k)$-space algorithm for \sconn runs on a graph $G$ and uses a sketch $\cS$ for \sconnt (which provably needs $\Omega(n)$ space) plus a sketch that recovers the non-edges of $G$ (which needs $\Omega(k)$ space for graphs that are $k$-close to tournaments). Again along the lines of our algorithm, it queries $\cS$ regarding the strong connectivity of certain completions of $G$ and tries to derive the solution for $G$ only from the answers to its queries. We prove that it then needs to make $\Omega(2^{\min(k,n)})$ queries, the bound attained by our algorithm. Formally, we establish the following theorem. 

\begin{restatable}{theorem}{sconnoraclelb}\label{thm:sconn-oracle-lb}
    Suppose we have an unknown digraph $G$ that is $k$-close to a tournament. Assume that we are given access to $F$, the set of $G$'s non-edges, and an oracle $\cO$ that, when queried with an orientation of edges corresponding to $F$, returns whether the resultant completion of $G$ is strongly connected or not. Then, detecting whether $G$ is strongly connected requires $\Omega(2^{\min(k,n)})$ queries to $\cO$.
\end{restatable}

To prove this, we will use reductions from the following problem in Boolean decision tree complexity. In this model, we have query access to the indices of a binary vector, and need to compute some function of it by making as few queries as possible (see the survey by \cite{BuhrmanW02} for formal definitions and other details). Consider an input vector $\bx \in
\{0,1\}^N$. Let $\bzero$ denote the all-zero vector 
and, for $i \in [N]$, let $\be_i$ denote the vector 
whose $i$th bit is $1$ and other bits are $0$. 
Then the problem $\uor_N$ is given by:
\[
  \uor_N(\bx) = \begin{cases}
    0 \,, & \text{if } \bx = \bzero \,, \\
    1 \,, & \text{if } \bx = \be_i \text{ for some $i$}\,,  \\
    \star \,, & \text{otherwise;}
  \end{cases}
  \]

Here, $\star$ denotes that any output is valid. Thus, the $\uor_N$ problem translates to the following: when $\bx$ is promised to be either $\bzero$ or $\be_i$ for some $i$, distinguish between the two cases. We shall use the following fact about the randomized (two-sided error) decision tree complexity $R^{\text{dt}}()$ of $\uor_N$. 

\begin{fact}[\cite{BuhrmanW02}] \label{fact:uor-lb}
  $R^{dt}(\uor_N)=\Omega(N)$
\end{fact}

We now prove \Cref{thm:sconn-oracle-lb}.

\begin{proof}[Proof of \Cref{thm:sconn-oracle-lb}]

 \begin{figure}[H]
        \centering
        
        \includegraphics[scale=0.8]{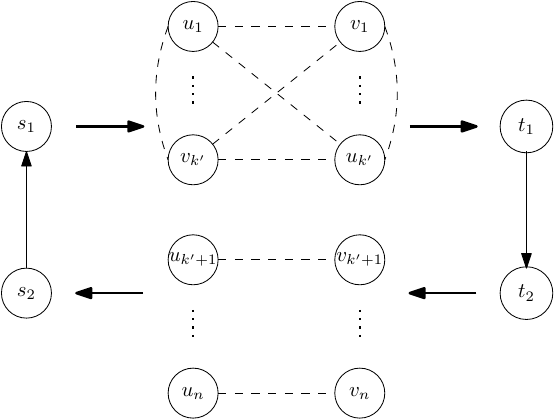}
        
        \caption{The graph $G_\bx$ constructed in the query lower bound for \sconn where bits 0, $k'+1$ and $n$ of the index $i$ are 0 and bit $k'$ of $i$ is 1. The dashed lines represent non-edges. Most edges are excluded for clarity, and go from left to right. Note that this is the construction for the $\Omega(2^n)$ case, but the $\Omega(2^k)$ case is very similar.}
        
        \label{fig:oracle-lb}
    \end{figure}
    
    We start by describing an $\Omega(2^k)$ query lower bound for $k \leq n$. We reduce from $\uor_{2^k}$. Let $\cA$ be an algorithm that detects strong connectivity of $G$ by making at most $q$ queries. We obtain a query algorithm for $\uor_{2^k}$ as follows. We run the algorithm $\cA$ on a graph $G = (V, E)$ with vertices $V = \{s_1, s_2, u_1, ..., u_n, v_1, ..., v_n, t_1, t_2\}$ and non-edges $F = \{\{u_j, v_j\} \: | \: j \in [k]\}$. Given an oracle $O'$ for \uor, we construct an oracle $\cO$ for strong connectivity by converting the orientations of edges in $F$ to an index query $i$ as follows. For all $j \in [k]$, if edge $j$ is oriented from $u_j$ to $v_j$ we set the $j$th bit of $i$ to 0, and if $j$ is oriented from $v_j$ to $u_j$ we set the $j$th but of $i$ to 1. We use $i$ as the index query to oracle $O'$. The oracle $\cO$ returns that the completion of $G$ is strongly connected if and only if $x_i = 0$. Finally, based on these queries, if $\mathcal{A}$ reports that the graph is strongly connected we output 0, otherwise we output 1.

    Let us show that for every vector $\bx$, (i) there is a graph $G_{\bx}$ such that the oracle $\cO$ is consistent on $G_\bx$, and (ii) our output for unique-or on input $\bx$, which is based on $\cA$'s output for $G_{\bx}$, is correct. The construction of $G_{\bx}$ will be defined in terms of two subsets of vertices, $L$ and $R$, defined in terms of $\bx$. If $\bx = \textbf{0}$, we add $u_j$ to $L$ and $v_j$ to $R$ for all $j \in [n]$. We then construct a graph $G_{\textbf{0}}$ with the following edges.

    \begin{itemize}
        \item $(s_2, s_1)$, $(s_1, \ell)$, $(\ell, s_2)$, $(s_1, r)$ and $(s_2, r)$ for all $\ell \in L$ and $r \in R$. 
        \item $(t_1, t_2)$, $(r, t_1)$, $(t_2, r)$, $(\ell, t_1)$ and $(\ell, t_2)$ for all $\ell \in L$ and $r \in R$.
        \item $(t_1, s_1)$ and arbitrary edges between all remaining pairs except those in $F$. \sahil{This can not be arbitrary when $\bx = \be_i$}
    \end{itemize}

    It is clear that $G_{\textbf{0}}$ is strongly connected since there is a path from $s_1$ to every vertex, a path from every vertex to $t_1$, and an edge from $t_1$ to $s_1$. Thus, irrespective of the orientation of the non-edges, we expect $\cO$ to return that the graph is strongly connected, which is satisfied since every query to $O'$ returns 0. This implies that $\cO$ is consistent on $G_{\textbf{0}}$. Moreover, note that $\mathcal{A}$ will correctly report that $G_{\textbf{0}}$ is strongly connected, so $\mathcal{A}$ will correctly output 0.

    Next, consider the case where $\bx = \mathbf{e}_i$. We construct $L$ and $R$ as follows. For all $j \in [k]$, if bit $j$ of $i$ is 0 we add $u_j$ to $L$ and $v_j$ to $R$, and if bit $j$ of $i$ is 1 we add $v_j$ to $L$ and $u_j$ to $R$. We then construct a graph $G_{\mathbf{e}_i}$ with the same edges as earlier, but with edge $(t_1, s_1)$ removed and edges $(s_1, t_1)$, $(s_2, t_1)$, $(s_1, t_2)$ and $(s_2, t_2)$ added. Note that $G_{\mathbf{e}_i}$ is no longer strongly connected, since there is no edges from vertices in $\{t_1, t_2\} \cup R$ to vertices in $\{s_1, s_2\} \cup L$. Moreover, note that the completion of $G_{\mathbf{e}_i}$ is not strongly connected if and only if every non-edge is oriented from $L$ to $R$, which is true if and only if the orientation of these edges matches the binary representation of the index $i$. This is precisely when $O'$ returns 1, so $\cO$ is consistent on $G_{\mathbf{e}_i}$. Moreover, note that $\mathcal{A}$ will correctly report that $G_{\textbf{0}}$ is not strongly connected, so $\mathcal{A}$ will correctly output 1. Thus, our reduction is correct.

    Finally, we know that $\uor_{2^k}$ has query complexity $\Omega(2^k)$ (\Cref{fact:uor-lb}) and our reduction makes $q$ queries to the oracle $O'$, so the query complexity for determining if $G$ is strongly connected must also be $\Omega(2^k)$ for $k \leq n$. 

    Next we show an $\Omega(2^n)$ query lower bound for $k \geq n$. We reduce from $\uor_{2^n}$. As before, we use an algorithm $\cA$ for strong connectivity with at most $q$ queries to design a query algorithm for \uor. The vertices of the graph remain the same, but the non-edges are now (see \Cref{fig:oracle-lb})
    $$F = \{\{u_j, v_j\} \: | \: j \in [n]\} \cup \{\{u_j, u_{j'}\}, \{v_j, v_{j'}\}, \{u_j, v_{j'}\} \: | \: j, j' \in [k']\}$$

    \noindent
    for some constant $k'$. Note that the number of missing edges is actually $\Theta(k'^2+n)$, which is $\Theta(k)$ for $k' = \Theta(\sqrt{k-n})$. Next, we construct the oracle $\cO$ as follows. We begin by constructing two sets of vertices $L$ and $R$. For every pair of vertices $u_j$ and $v_j$, if the edge in the query is oriented from $u_j$ to $v_j$ we add $u_j$ to $L$ and $v_j$ to $R$, and if the edge is oriented from $v_j$ to $u_j$ we add $v_j$ to $L$ and $u_j$ to $R$. If there is any edge in the query that goes from a vertex in $R$ to a vertex in $L$, we return that the graph is strongly connected. If there is no such edge we construct the index $i$ based on the edges as earlier, and have $\cO$ return that the graph is strongly connected if $x_i = 0$ and that the graph is not strongly connected if $x_i = 1$. Finally, based on these queries, if $\mathcal{A}$ reports that the graph is strongly connected we output 0, otherwise we output 1.

    Our construction of $G_{\bx}$ will be the same as earlier, excluding all the edges in $F$. Note that for $\bx = \textbf{0}$, $G_{\textbf{0}}$ is still strongly connected, so the proof follows. Thus, let us consider $\bx = \mathbf{e}_i$. As before, $G_{\mathbf{e}_i}$ is not strongly connected, and remains so if and only if every non-edge between $L$ and $R$ is oriented from $L$ to $R$. This holds if and only if (i) the orientation of edges $(u_j, v_j)$ matches the binary representation of index $i$ and (ii) the orientation of edges $(u_j, v_k)$ follows this orientation. This is precisely when $O'$ returns 1, so $\cO$ is consistent on $G_{\mathbf{e}_i}$. Thus, our reduction remains correct.

    Finally, we know that $\uor_{2^n}$ has query complexity $\Omega(2^n)$ (\Cref{fact:uor-lb}) and our reduction makes at most $q$ queries to the oracle $O'$, so the query complexity for determining if $G$ is strongly connected must also be $\Omega(2^n)$ for $k > n$. Combining the two bounds, we see that deciding whether $G$ is strongly connected requires $\Omega(2^{\min\{k, n\}})$ queries.  
 \end{proof}

\subsection{Other Applications of SCC Decomposition of Tournaments}\label{sec:otherapps}
  
\subsubsection{Hamiltonian Cycle}\label{subsec:hamcyc}

Using our SCC framework, we obtain the first streaming algorithm for \hcyct (see \Cref{sec:results} for formal definition). We build on the parallel algorithm for \hcyct in \cite{SorokerHamCycle} and adopt it in the streaming setting. Our algorithm is given by \Cref{alg:cycle}. There we use our \sccdect algorithm (\Cref{alg:sccgraph}) as a subroutine. 

\begin{restatable}{theorem}{hamcyc}\label{thm:hamcyc}
    There is a deterministic $O(\log n)$-pass $\tO(n)$-space algorithm for \hcyct.
\end{restatable}

 \begin{proof}\Cref{alg:cycle} describes the algorithm for finding a Hamiltonian cycle in an input tournament graph $G$, with \Cref{alg:restrictedpath} as a subroutine.

\begin{algorithm}[!ht]
  \caption{\:\texttt{Cycle}$(G)$: Finds a Hamiltonian cycle in $G = (V, E)$.
    \label{alg:cycle}}
  
  \begin{algorithmic}[1]

    \Statex \underline{\textbf{Initialize}:}
    \State $n \gets |V|$
    \If{$n = 1$}
    \State Return $v \in V$
    \EndIf

    \Statex

    \Statex \underline{\textbf{First pass}:}
    \State Find a vertex $v \in V$ such that $d_{\text{in}}(v) \leq \lfloor n/4 \rfloor$ and $d_{\text{out}}(v) \leq \lfloor n/4 \rfloor$

    \Statex
    \Statex \underline{\textbf{Second pass}:}
    \State Let $L(v) = \{u \: : \: (u, v) \in E\}$ and $W(v) = \{w \: : \: (v, w) \in E\}$
    \State Find SCCs $\langle L_1, \ldots, L_q \rangle$ of $L(v)$ and $\langle W_p, \ldots, W_1\rangle$ of $W(v)$

    \Statex
    \Statex \underline{\textbf{Third pass}:}
    \State $m \gets 0$, $k\gets 0$
    \For{$(u, v) \in E$}
        \If{$u \in W_1$ and $v \in L_i$ and $i > m$}
            \State $m \gets i$, $w_1 \gets u$, $l_1 \gets v$
        \EndIf
        \If{$u \in W_i$ and $v \in L_1$ and $i > k$}
            \State $k \gets i$, $w_2 \gets u$, $l_2 \gets v$
        \EndIf
    \EndFor

    \Statex
    \Statex \underline{\textbf{$O(\log n)$ passes (parallel)}:}
    \State $H_1 \gets \texttt{RestrictedPath}(W_1, \text{`end'}, w_1)$
    \State $H_2 \gets \texttt{RestrictedPath}(\{L_m, \ldots, L_q\}, \text{`start'}, w_1)$
    \State $H_3 \gets \texttt{RestrictedPath}(\{w_k, \ldots, W_p\}, \text{`end'}, w_2)$
    \State $H_4 \gets \texttt{RestrictedPath}(L_1, \text{`start'}, l_2)$
    \State $H_5 \gets \texttt{Path}(\{W_2, \ldots, W_{k-1}, L_2, \ldots, L_{m-1}\})$

    \Statex

    \Statex \underline{\textbf{Post-processing}:}
    \State Return $(v, H_1, H_2, H_3, H_4, H_5, v)$
  \end{algorithmic}
  \end{algorithm}

  \begin{algorithm}[!ht]
  \caption{\:\texttt{RestrictedPath}$(G, t, u)$: Finds a Hamiltonian path in $G = (V, E)$. If $t = \text{`start'}$ the path will start at $u$, and if $t = \text{`end'}$ then the path will end at $u$.
    \label{alg:restrictedpath}}
  
  \begin{algorithmic}[1]

    \Statex \underline{\textbf{Initialize}:}
    \State $n \gets |V|$
    \If{$n = 1$}
    \State Return $v \in V$
    \EndIf

    \Statex
    \Statex \underline{\textbf{First pass}:}
    \State Find SCCs $\langle C_1, \ldots, C_k \rangle$ of $G$

    \Statex

    \Statex \underline{\textbf{$O(\log n)$ passes (parallel)}:}
    \If{$t = \text{`start'}$}
    \State $H_1 \gets \texttt{Cycle}(C_1)$
    \State $H_2 \gets \texttt{Path}(\{C_2, \ldots, C_k\})$
    \State Remove the edge into $u$ from $H_1$
    \EndIf

    \If{$t = \text{`end'}$}
    \State $H_1 \gets \texttt{Path}(\{C_1, \ldots, C_{k-1}\})$
    \State $H_2 \gets \texttt{Cycle}(C_k)$
    \State Remove the edge out of $u$ from $H_2$
    \EndIf

    \Statex
    \Statex \underline{\textbf{Post-processing}:}
    \State Return $(H_1, H_2)$
  \end{algorithmic}
  
\end{algorithm}

 The subroutine $\texttt{Path}(G)$, referenced in both algorithms, finds a Hamiltonian path of $G$, and can be implemented using the algorithm in either \cite{ChakrabartiGMV20} or \cite{BawejaJW22}. This subroutine requires $\log n$ passes for semi-streaming space.

The proof of correctness follows directly from \cite{SorokerHamCycle}, which proposes a similar algorithm for finding a Hamiltonian path on tournaments in the setting of the parallel algorithms. Thus, we only show that the algorithm requires $O(\log n)$ passes. Let us start by considering $\texttt{Cycle}(G)$ on a tournament graph $G$ with $n$ vertices. The algorithm requires 3 initial passes, first to find a `mediocre' vertex, then to find the corresponding SCC graphs using our single-pass algorithm, and finally to find vertices $w_1$ and $w_2$ with the particular property required by the algorithm. Following this, we make 4 calls to the $\texttt{RestrictedPath}$ subroutine and 1 call to the $\texttt{Path}$ subroutine, each of which is run on a disjoint subgraph of $G$ and can therefore be run in parallel. We know that $\texttt{Path}$ runs in $\log n$ passes. Moreover, note that each of the subgraphs passed to $\texttt{RestrictedPath}$ have size at most $\frac{3}{4} n$ vertices, as shown in \cite{SorokerHamCycle}. Thus, the number of passes required by the algorithm is $O(\max(\log n, p))$, where $p$ is the number of passes required by $\texttt{RestrictedPath}$ on a graph of size $\frac{3}{4} n$.

Next, let us consider $\texttt{RestrictedPath}(G, t, u)$, where $G$ is an arbitrary tournament graph with $n$ vertices and $t$ and $u$ are arbitrary inputs. The algorithm requires an initial pass find the SCC graph, followed by 1 call to $\texttt{Cycle}$ and 1 call to $\texttt{Path}$. As before, we know that $\texttt{Path}$ requires $\log n$ passes. We also see that the subgraph passed to $\texttt{Cycle}$ can have up to $n$ vertices. Thus, the number of passes required by the algorithm is $O(\max(\log n, p))$, where $p$ is the number of passes required by $\texttt{Cycle}$ on a graph of size $n$.

Finally, note that for each pair of calls to $\texttt{Cycle}$ and $\texttt{RestrictedPath}$, the size of the subgraph being processed decreases by at least $\frac{3}{4}$. Thus, it is clear that there are at most $O(\log n)$ calls to both $\texttt{Cycle}$ and $\texttt{RestrictedPath}$, each of which require a constant number of initial passes, so our algorithm succeeds in $O(\log n)$ passes. 

\end{proof}

 In addition, we observe that if one used \cite{BawejaJW22}'s algorithm as a subroutine, then it would only yield an $O(\log^2 n)$-pass algorithm in the worst case. This is because we know that, in the worst case, a single vertex $v$ may be involved in $O(\log n)$ calls to $\texttt{Cycle}$, where in the $i$th call the size of the subgraph is $\left(\frac{2}{3}\right)^i \cdot n$. Thus, we must run the SCC algorithm on each of these subgraphs, so the number of passes required is
$$\sum_{i=0}^{c\log n} \log\left(\left(\frac{2}{3}\right)^i \cdot n\right) = \Omega(\log^2(n)).$$

\subsubsection{Hamiltonian Path and Feedback Arc Set}\label{subsec:smallscc}

Here, we show applications of \sccdect to design algorithms for \hpatht and \fast.

\begin{restatable}{theorem}{hampathfas}\label{thm:hampathfas}
    Given an input $n$-node tournament $T$ whose largest SCC has size $s$, for any $p\geq 1$, there are deterministic $(p+1)$-pass $\tO(ns^{1/p})$ space streaming algorithms for \hpatht and $(1+\eps)$-approximate \fast. In particular, there are $O(\log s)$-pass semi-streaming algorithm for each problem.      
\end{restatable}

 \cite{BawejaJW22} gave $p$-pass $\tO(n^{1+1/p})$-space algorithms for each of the above problems. Our algorithm uses $\tO(ns^{1/(p-1)})$ space in $p$ passes. Thus, we see that we have improved $p$-pass algorithms whenever $s < n^{(p-1)/p}$. 
 
\begin{proof}[Proof of \Cref{thm:hampathfas}]
    We begin by giving an improved algorithm for finding a Hamiltonian path by merging the $p$-pass $\tO(n^{1+1/p})$ space algorithm from \cite{BawejaJW22} and our single-pass SCC decomposition algorithm, which we will refer to as the original algorithm. Given a tournament graph $T = (V, E)$, in the first pass we run the SCC decomposition algorithm, from which we obtain the SCC graph $\langle V_1, \ldots, V_\ell \rangle$ of $T$. We then run the original algorithm on each component $V_1, \ldots, V_\ell$ in parallel, giving us Hamiltonian paths $H_1, \ldots, H_\ell$. We return $(H_1, \ldots, H_\ell)$, the concatenation of the individual Hamiltonian paths.

    Let us start by showing that our algorithm returns a Hamiltonian path of $T$. We know that for all $i \in [\ell]$, $H_i$ is a Hamiltonian path of $V_i$. By \Cref{lem:scc-cut-orient}, we know that for all $i \in [\ell-1]$, every edge between $V_i$ and $V_{i+1}$ is oriented from $V_i$ to $V_{i+1}$. Thus, $T$ contains an edge from the last vertex in $H_i$ to the first vertex in $H_{i+1}$, so $(H_1, \ldots, H_\ell)$ is a Hamiltonian path of $T$.

    Next, let us show that given $(p+1)$ passes, our algorithm requires $\tO(ns^{1/p})$ space. We use a single pass to compute the SCC decomposition in semi-streaming space, giving us $p$ passes to compute the Hamiltonian paths of SCCs $V_1, \ldots, V_\ell$. Let us begin by dividing the SCCs into groups $G_i = \{V_j \: | \: 2^i \leq |V_j| < 2^{i+1} \} $. For all $i$, it is clear that $|G_i| \leq \frac{n}{2^i}$. Moreover, we see that in $p$ passes, computing the Hamiltonian path of an SCC $V_j \in G_i$ requires at most $2^{(i+1)(1+1/p)}$. Thus, the total space required to compute the Hamiltonian paths of SCCs in $G_i$ is at most $\frac{n}{2^i} \cdot 2^{(i+1)(1+1/p)} = O(n \cdot 2^{i/p})$. Moreover, note that the size of the largest SCC is $s$, so there are at most $\log s$ non-empty groups $G_i$. Thus, we see that the total space required for the algorithm is at most
    $$\sum_{i=1}^{\log s} n \cdot 2^{i/p} = \tO(ns^{1/p}).$$

    We can use a very similar modification to obtain a $(1+\eps)$-approximation algorithm for FAS. In this case, the original algorithm is the $p$-pass $\tO(n^{1+1/p})$ space algorithm for FAS-T from \cite{BawejaJW22}. Given a tournament graph $T = (V, E)$, in the first pass we run the SCC decomposition algorithm, from which we obtain the SCC graph $\langle V_1, \ldots, V_\ell \rangle$ of $T$. We then run the original algorithm on each component $V_1, \ldots, V_\ell$ in parallel, giving us FAS orderings $\sigma_1, \ldots, \sigma_\ell$. We return $(\sigma_1, \ldots, \sigma_\ell)$, the concatenation of the individual FAS orderings.

    Let us start by showing that our algorithm returns a $(1+\eps)$-approximation for the FAS ordering of $T$. By \Cref{lem:scc-cut-orient}, we know that for all $i, j \in [\ell]$ such that $i < j$, every edge between $V_i$ and $V_j$ is oriented from $V_i$ to $V_j$. Thus, there are no back-edges between components in the ordering $(\sigma_1, \ldots, \sigma_\ell)$. This also implies that the optimal FAS ordering of $T$ is the concatenation of the optimal FAS orderings of $V_1, \ldots, V_\ell$. Thus, since each $\sigma_i$ is a $(1+\eps)$-approximation of the optimal FAS ordering of $V_i$, $(\sigma_1, \ldots, \sigma_\ell)$ is a $(1+\eps)$-approximation of the optimal FAS ordering of $T$. 

    The proof that our algorithm requires $\tO(ns^{1/p})$ space in $(p+1)$ passes follows directly from the proof for Hamiltonian path provided above, and is therefore not repeated.
\end{proof}

        \section{Lower Bounds for Feedback Arc Set and Shortest Distance on Tournaments}

In this section, we investigate problems that remain hard on tournaments. In particular, we show that exactly solving \fast, \fassizt, or \stdist requires storage of almost the entire graph.  We also give a generalized space-approximation tradeoff for \fast. 

    \subsection{Warm Up: A Lower Bound for Exact \fast}
    
    First we prove a lower bound for the exact version.

\begin{restatable}{lemma}{exactfast}\label{lem:exactfast}
     Any randomized single-pass algorithm that exactly solves \fast requires $\Omega(n^2)$ space.
\end{restatable}

\begin{proof}
\begin{figure}[H]
        \centering
        
        \includegraphics[scale=0.5]{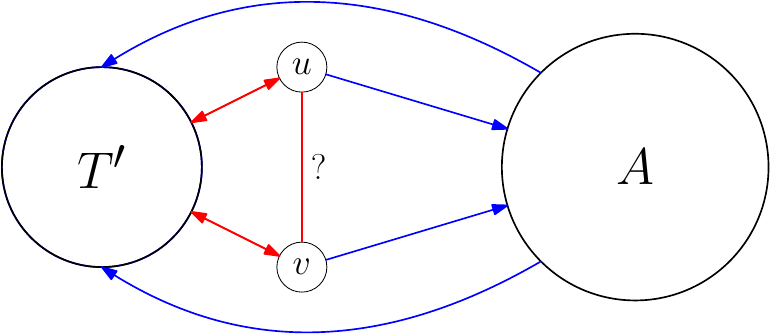}
        
        \caption{The graph constructed in the lower bound for \fast. Edges Alice adds are represented in red, and edges Bob adds are represented in blue. The `$?$' marks the edge Alice and Bob are trying to determine.}
        
        \label{fig:fas-t-lb}
    \end{figure}
        We reduce from $\tedge_N$. Alice and Bob construct a tournament $G = (V, E)$ with $n = 3N$ vertices (see \Cref{fig:fas-t-lb}). Alice, based on her input, initially constructs the tournament $T$ with $N$ vertices. Let $u$ and $v$ be the vertices between which Bob needs to identify the orientation of the edge, and let $T'$ be the subgraph of $T$ defined by the vertex set $V - \{u, v\}$. Bob creates a new acyclic tournament, $A$, on the remaining $2N$ vertices. He then adds edges from $u$ and $v$ to every vertex in $A$, and from every vertex in $A$ to every vertex in $T'$. This completes the construction of the graph $G$, and it is easy to check that $G$ is a tournament. Given the minimum FAS ordering of $G$, Bob will announce that the edge in $G$ is oriented from $u$ to $v$ if $u$ is before $v$ in the ordering, and from $v$ to $u$ otherwise. 

        Let us show that $G$ contains an edge oriented from $u$ to $v$ if and only if $u$ is before $v$ in any minimum FAS ordering. To do so, we will first show that if $G$ contains the edge $(u, v)$, then any minimum FAS ordering of $G$ is of the form $\pi = (u, v, A, T')$, where $A$ and $T'$ are ordered by some minimum FAS ordering of their induced subgraphs. To begin, note that $\text{FAS-SIZE}(\pi) = \delta(T', u) + \delta(T', v) + \text{MIN-FAS-SIZE}(T')$, where $\delta(T', u)$ and $\delta(T', v)$ are the number of edges from $T'$ to $u$ and $v$ respectively. We want to show that every other ordering of vertices has a larger FAS size. 
        
        Let $\sigma$ be an arbitrary ordering of vertices in $G$, and let $i$ and $j$ be the number of vertices of $T'$ that appear before $u$ and $v$ in $\sigma$ respectively. Assume that $i \geq 1$, and let $t$ be the left-most vertex of $T'$ before $u$ in $\sigma$. Consider an arbitrary vertex $a \in A$. If $a$ appears before $u$ in $\sigma$, it must contribute at least 1 to the FAS of $\sigma$, since the edge from $u$ to $a$ is a back-edge. Similarly, if $a$ is after $t$ in $\sigma$, it must contribute at least 1 to the FAS of $\sigma$, since the edge from $a$ to $t$ is a back-edge. Moreover, since $t$ is before $u$ in $\sigma$, $a$ must be either before $u$ or after $t$ (or both). Thus, every vertex $a \in A$ contributes at least 1 edge to the FAS of $\sigma$, so $A$ contributes at least $2N$ edges to the FAS of $\sigma$. Moreover, note that the edges within $T'$ must contribute at least $\text{MIN-FAS-SIZE}(T')$ edges to the FAS of $\sigma$, irrespective of how the vertices of $T'$ are ordered. Thus, the FAS size of $\sigma$ is at least $\text{MIN-FAS-SIZE}(T') + 2N$. Note, however, that $\pi$ has an FAS size of at most $\text{MIN-FAS-SIZE}(T') + 2N - 2$, since both $\delta(T', u)$ and $\delta(T', v)$ are at most $N-1$. Thus, if $i \geq 1$ we know that $\text{FAS}(\sigma) > \text{FAS}(\pi)$. By a symmetric argument, we see that the same applies if $j \geq 1$. Thus, the only way we can get an ordering with FAS less than or equal to that of $\pi$ is if $i = j = 0$.

        Since $i = j = 0$, we know that $u$ and $v$ must be before $T'$ in $\sigma$. Thus, it is clear that the minimum FAS is at least $\delta(T', u) + \delta(T', v) + \text{MIN-FAS-SIZE}(T')$, which matches $\text{FAS}(\pi)$. Moreover, note that this is the unique minimum FAS up to rearranging the vertices in $T'$, since moving any of the vertices in $A$ will always add edges to the FAS, and if $G$ contains the edge $(u, v)$ then swapping $u$ and $v$ in $\pi$ adds an edge to the FAS. Thus, $\pi$ is the minimum FAS ordering, and is unique up to rearranging the vertices in $T'$. By symmetry, we also see that if $G$ contains the edge $(v, u)$ then $\pi' = (v, u, A, T')$ is the minimum FAS ordering, and is unique up to rearranging the vertices in $T'$. Thus, we see that $G$ contains an edge oriented from $u$ to $v$ if and only if $u$ is before $v$ in the minimum FAS ordering.

        This implies that we can use a streaming algorithm for \fassizt on $n = 3N$ vertices to solve $\tedge_N$. Since $\tedge_N$ requires $\Omega(N^2)$ bits of one-way communication (see \ref{prop:tedge-lb}), \fassizt must require $\Omega(n^2)$ space in a single pass.
    \end{proof}

\subsection{A Generalized Lower Bound for Approximate \fast}

Here we establish the following generalization.

  \begin{restatable}{theorem}{fasmain}\label{thm:fas-main}
   Given any $\eps>0$, a single-pass $(1+\eps)$-approximation algorithm for \fast needs $\Omega(\min\{n^2,n/\sqrt{\eps}\})$ space. 
\end{restatable}

   To prove this, we define a relaxed version of direct sum \cite{chakrabartiswy01} and prove a lower bound for it. This is slightly stronger than the usual direct sum lower bound, and can be of independent interest. First recall the standard direct sum problem.

   \begin{definition}[Direct sum \cite{chakrabartiswy01}]
    Let $f$ be a function such that in the one-way communication model, Alice, holding input $x$, needs to send $\Omega(s)$ bits to Bob who holds input $y$, so that he can output $f(x,y)$ with probability at least $2/3$. Alice and Bob are given $k$ independent instances $(x_1,y_1),\ldots, (x_k,y_k)$ and they need to compute $f$ on \emph{each} of these instances with probability at least $2/3$.  
\end{definition}

Now consider the following relaxed version.

    \begin{definition} [Relaxed direct sum] 
        Let $f$ be a function such that in the one-way communication model, Alice, holding input $x$, needs to send $\Omega(s)$ bits to Bob who holds input $y$, so that he can output $f(x,y)$ with probability at least $p$. Alice and Bob are given $k$ independent instances $(x_1,y_1),\ldots, (x_k,y_k)$ and they need to compute $f$ on \emph{at least $9/10$} of these instances with probability at least $10p/9$.  
    \end{definition}

    \begin{lemma}
    The relaxed direct sum problem requires $\Omega(ks)$ communication for $p=2/3$.
    \end{lemma}\label{lem:rel-dir-sum}

    \begin{proof}
        Given a communication protocol $\Pi$ for relaxed direct sum, we can use it to solve direct sum as follows. Given $k$ copies of $f$, where Alice is given the bits $A_1, \ldots, A_{k}$ and Bob is given the bits $B_1, \ldots, B_{k}$, they first use shared randomness to generate a random permutation $\sigma$ of $(1, \ldots, k)$. Next, they will run $\Pi$ on $A_{\sigma(1)}, \ldots, A_{\sigma(k)}$ and $B_{\sigma(1)}, \ldots, B_{\sigma(k)}$. Let $x_{\sigma(1)}, \ldots, x_{\sigma(k)}$ be the solutions computed by $\Pi$. Alice and Bob will output $x_1, \ldots, x_k$ as the solution to the direct sum problem.

        Let us show that every instance $(A_i, B_i)$ in the direct sum problem is solved with probability at least $p$. Let $\Pr [\Pi(j)]$ be the probability that the protocol $\Pi$ solves instance $j$ of the relaxed direct sum. Note that $\sum_{j \in [k]} \Pr [\Pi(j)]$ is equal to the expected number of instances that $\Pi$ solves. Thus, 
        $$\sum_{j \in [k]} \Pr [\Pi(j)] \geq \frac{9k}{10} \cdot \frac{10p}{9} = k \cdot p.$$

        By the law of total probability, we know that instance $(A_i, B_i)$ in the direct sum problem is solved correctly with probability $\sum_{j \in [k]} \Pr[\sigma(i) = j] \cdot \Pr[\Pi(j)]$. We know that $\Pr[\sigma(i) = j] = \frac{1}{k}$ for all $i$ and $j$. Thus, the probability that instance $(A_i, B_i)$ is solved correctly simplifies to $\frac{1}{k} \cdot \sum_{j \in [k]} \Pr [\Pi(j)] \geq \frac{1}{k} \cdot k \cdot p = p$. Thus, we see that every instance $(A_i, B_i)$ in the direct sum problem is solved with probability at least $p$. We know that the direct sum problem requires $\Omega(ks)$ bits of communication for $p=2/3$ \cite{chakrabartiswy01}, so the relaxed direct sum problem must also require $\Omega(ks)$ bits of communication.
    \end{proof}

    \begin{corollary} \label{cor:tedge-weak-dir-sum}
        Given $k$ independent copies of $\tedge_N$, solving $9/10$ of them with probability at least $2/3$ each requires $\Omega(kN)$ communication. 
    \end{corollary}

    \begin{proof}
        Follows from \Cref{lem:rel-dir-sum} and \Cref{prop:tedge-lb}.
    \end{proof}

We are now ready to prove \Cref{thm:fas-main}

    \begin{proof}[Proof of \Cref{thm:fas-main}]
        We reduce from the relaxed direct sum of $N = n/k$ independent copies of $\tedge_k$ (see \Cref{cor:tedge-weak-dir-sum}). The reduction involves creating $n/k$ copies of the construction in \Cref{lem:exactfast} to form a single larger tournament. In particular, Alice and Bob construct a tournament $G = (V, E)$ with $3n$ vertices. Alice, based on her input tournaments, constructs $N$ disjoint tournaments $T_1, \ldots, T_{N}$ in $G$, each with $k$ vertices. Let $T'_i$ be the subgraph of $T_i$ defined by the vertex set $T_i-\{u_i, v_i\}$ for all $i \in [N]$. Bob constructs $N$ disjoint acyclic tournaments $A_1, \ldots, A_{N}$, each with $2k$ vertices. He then adds edges from $u_i$ and $v_i$ to all vertices in $A_i$ and from all vertices in $A_i$ to all vertices in $T'_i$ for all $i \in [N]$. Finally, he adds edges from all vertices in $\{T_i, A_i\}$ to all vertices in $\{T_j, A_j\}$ for all $1 \leq i < j \leq N$. This completes the construction of the graph $G$, and it is easy to check that $G$ is a tournament. Given the approximate minimum FAS ordering $\pi$ of $G$, for every $i \in [N]$, Bob will announce that the edge in $T_i$ is oriented from $u_i$ to $v_i$ if $\pi(u_i) < \pi(v_i)$, and from $v_i$ to $u_i$ otherwise. 

        Let us show that a $(1+\epsilon)$ approximation algorithm for FAS solves at least a $9/10$ fraction of the \tedge instances correctly in the reduction above, where $\epsilon = 1/(10k^2)$. First, as shown in \Cref{lem:exactfast}, note that the optimal FAS ordering for a given sub-tournament $T_i$ is $\pi_i = (u_i, v_i, A_i, T'_i)$ if $T_i$ contains an edge from $u_i$ to $v_i$ and $\pi_i = (v_i, u_i, A_i, T'_i)$ otherwise, where $A_i$ and $T'_i$ are ordered by some minimum FAS ordering of their induced subgraphs. Moreover, ordering the sub-tournaments $T_i$ by index ensures that there are no back edges from one sub-tournament to another. Thus, the optimal FAS ordering for the complete tournament $G$ is of the form $\pi = (u_1, v_1, A_1, T'_1, \ldots, u_{N}, v_{N}, A_{N}, T'_{N})$, where $u_i$ and $v_i$ are swapped if $T_i$ contains an edge from $v_i$ to $u_i$. 

        Next, note that the size of the minimum FAS of $G$ is at most $Nk^2$. This is because there are no back edges between sub-tournaments, each sub-tournament has $k^2$ edges, and there are $N$ sub-tournaments. Thus, a $(1+\epsilon)$ approximation can have at most $\epsilon \cdot Nk^2 = N/10$ more back edges than the minimum FAS ordering. 

        Let $\sigma$ be the FAS ordering returned by the $(1+\epsilon)$ approximation algorithm. If $G$ contains an edge from $u$ to $v$ but $\sigma(u) > \sigma(v)$, or if $G$ contains an edge from $v$ to $u$ but $\sigma(v) > \sigma(u)$, then $\sigma$ contains a back-edge that $\pi$ does not, but the optimal ordering of the rest of the graph remains the same. Thus, every such inversion adds 1 to the FAS size of $\sigma$ as compared to that of $\pi$, so the number of such inversions in $\sigma$ is at most $N/10$. This implies that $9N/10$ of the sub-tournaments $T_i$ are in the same order in $\pi$ as in optimal FAS ordering $\sigma$, so as shown in \Cref{lem:exactfast}, the edges between $u$ and $v$ in these tournaments are correctly identified. Thus, at least a $9/10$ fraction of the \tedge instances are solved.

        This implies that we can use a $(1+\epsilon)$ approximation streaming algorithm for \fast on $n = Nk$ vertices to solve the relaxed direct sum of $N$ independent copies of $\tedge_k$. Since this requires $\Omega(Nk^2)$ bits of one-way communication (see \Cref{cor:tedge-weak-dir-sum}), the \fast approximation algorithm must require $\Omega(nk) = \Omega(n/\sqrt{\epsilon})$ space in a single pass.
    \end{proof}

Next, we show that solving exact \fassizt also requires $\Omega(n^2)$ space in a single pass. Note that, given an FAS ordering, we can find the corresponding number of back-edges in one more pass, but a single-pass lower bound for \fassizt does not imply the same lower bound for \fas.

\begin{restatable}{theorem}{fastsizelb}\label{thm:fastsize-lb}
       Solving \fassizt exactly in a single pass requires $\Omega(n^2)$ space.
   \end{restatable}

 \begin{proof}%[Proof of \Cref{thm:fastsize-lb}]

    \begin{figure}[H]
        \centering
        
        \includegraphics[scale=0.6]{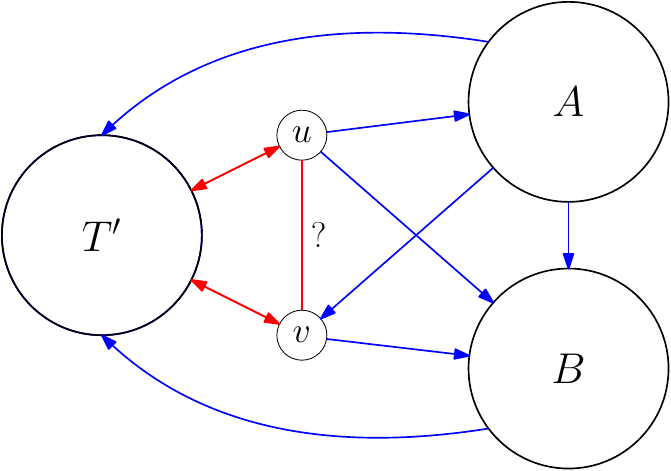}
        
        \caption{The graph $G_1$ constructed in the lower bound for \fassizt. Edges Alice adds are represented in red, and edges Bob adds are represented in blue. The `$?$' marks the edge Alice and Bob are trying to determine. Note that, in $G_2$, $u$ and $v$ are swapped.}
        
        \label{fig:fassiz-t-lb}
    \end{figure}
    
        We reduce from $\tedge_N$ using 2 instances of the streaming algorithm. Alice and Bob construct 2 tournaments $G_1 = (V_1, E_1)$ and $G_2 = (V_2, E_2)$ each with $5N$ vertices. In both $G_1$ and $G_2$, Alice, based on her input, constructs tournaments $T_1$ and $T_2$ with $N$ vertices each. Let $u$ and $v$ be the vertices between which Bob needs to identify the orientation of the edge, and let $T'$ be the subgraph of $T$ defined by the vertex set $V - \{u, v\}$. In $G_1$, Bob creates 2 new acyclic tournaments, $A_1$ and $B_1$, each with $2N$ vertices. He then adds edges from $u$ to $A_1$ and $B_1$, from $A_1$ to $v$, $B_1$ and $T'$, from $v$ to $B_1$, and from $B_1$ to $T'$. Similarly, in $G_2$, Bob creates 2 new acyclic tournaments, $A_2$ and $B_2$, each with $2N$ vertices. He then adds edges from $v$ to $A_1$ and $B_1$, from $A_1$ to $u$, $B_1$ and $T'$, from $u$ to $B_1$, and from $B_1$ to $T'$. Note that the only change in the two instances is that $u$ and $v$ are swapped. It is easy to verify that both $G_1$ and $G_2$ are tournaments. See \Cref{fig:fassiz-t-lb} for an illustration of $G_1$. Given the minimum FAS size of both $G_1$ and $G_2$, Bob can announce that the edge in $G$ is oriented from $u$ to $v$ if $G_1$ has a smaller FAS size than $G_2$, and from $v$ to $u$ otherwise. 
    
        Let us start by showing that $\pi_1 = \{u, A_1, v, B_1, T'\}$ is a minimum FAS ordering of $G_1$, where $A_1$, $B_1$ and $T'$ are ordered by a minimum FAS ordering of their subgraphs. The proof of this is very similar to that of \Cref{lem:exactfast}, and will borrow some notation. To begin, note that $\text{FAS}(\pi_1) \leq \delta(T', u) + \delta(T', v) + \text{MIN-FAS-SIZE}(T') + 1$, where equality holds when $G$ contains an edge oriented from $v$ to $u$. We want to show that every other ordering of vertices has at least this FAS size.

        First, note that as long as $u$ and $v$ are before $T'$ in the ordering, the FAS size is at least $\delta(T', u) + \delta(T', v) + \text{MIN-FAS-SIZE}(T')$. Thus, the only way to get a smaller FAS size than $\pi_1$ is by placing $v$ before $u$ in the case where $G$ contains an edge oriented from $v$ to $u$. However, this forces every vertex of $A_1$ to be either before $u$ or after $v$ in the ordering, adding a total of at least $4N$ back edges. Thus, placing $v$ before $u$ in this ordering gives us an FAS size of at least $\delta(T', u) + \delta(T', v) + \text{MIN-FAS-SIZE}(T') + 4N$, which is too large.

        This implies that the only way to reduce the FAS size is at least one of  $u$ and $v$ comes after some vertices of $T'$ in the ordering. However, note that the subgraph of $G$ formed by the vertices $u$, $v$, $B$ and $T'$ is isomorphic to the construction in \Cref{lem:exactfast}. Thus, we know that if at least one of  $u$ and $v$ comes after some vertices of $T'$ in the ordering, the FAS of the ordering is at least $\text{MIN-FAS-SIZE}(T') + 2N$, which is too large. Thus, we can conclude that $\pi_1$ is a minimum FAS ordering of $G_1$. By symmetry, we also see that $\pi_2 = \{v, A_2, u, B_2, T'\}$ is a minimum FAS ordering of $G_2$. This implies that if the edge in $G$ is oriented from $u$ to $v$, the minimum FAS size of $G_1$ is $\delta(T', u) + \delta(T', v) + \text{MIN-FAS-SIZE}(T')$ and the minimum FAS size of $G_2$ is $\delta(T', u) + \delta(T', v) + \text{MIN-FAS-SIZE}(T') + 1$. The opposite holds if the edge in $G$ is oriented from $v$ to $u$. Thus, Bob will correctly identify the orientation of the edge between $u$ and $v$ in $G$. 

        Finally, let us see how we can use a streaming algorithm $\mathcal{A}$ for \fassizt to solve $\idx_N$. As stated above, Alice will construct two streams, one for $G_1$ and for for $G_2$, and run an instance of $\mathcal{A}$ on each. She will then send the memory state of both instances to Bob, who will continue the construction of $G_1$ and $G_2$ as described above to obtain the minimum FAS size of each. Finally, he will compare the two values; if $G_1$ has a smaller minimum FAS size than $G_2$, he will declare that the edge in $T$, Alice's input, is oriented from $u$ to $v$, and if $G_2$ has the smaller minimum FAS size he will declare that the edge is oriented from $v$ to $u$. Thus, we can use 2 copies of a streaming algorithm for \fassizt on $n = 5N$ vertices to solve $\tedge_N$. Since $\tedge_N$ requires $\Omega(N^2)$ bits of one-way communication (see \Cref{prop:tedge-lb}), $\fassizt$ must require $\Omega(n^2)$ space in a single pass.  
    \end{proof}

\subsection{Shortest Distance}

We proved that NP-hard problems like \fas and \fast need $\Omega(n^2)$ space in a single pass. Now we prove that classical polynomial-time solvable problems such as \stdist also need $\Omega(n^2)$ space on tournaments.

 \begin{restatable}{theorem}{stdistthm}\label{thm:stdist}
    Any randomized single-pass algorithm that exactly computes the distance between two input nodes in a tournament requires $\Omega(n^2)$ space.
\end{restatable}

\begin{proof}%[Proof of \Cref{thm:stdist}]

    \begin{figure}[H]
        \centering
        
        \includegraphics[scale=0.7]{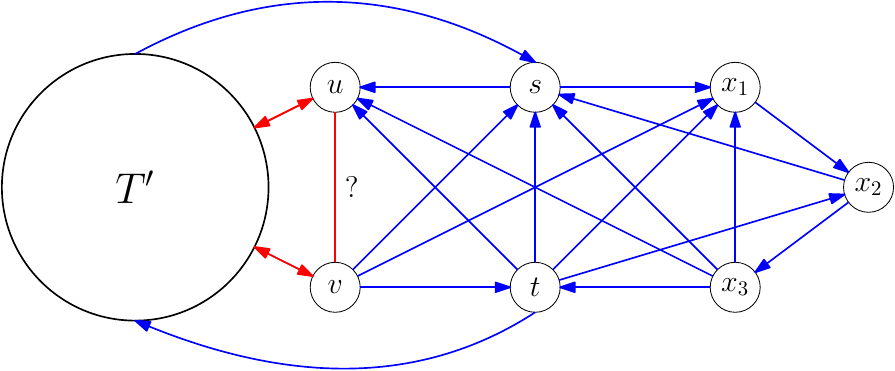}
        
        \caption{The graph constructed in the lower bound for \stdistt. Edges Alice adds are represented in red, and edges Bob adds are represented in blue. Edges independent of the input have been excluded for clarity. The `$?$' marks the edge Alice and Bob are trying to determine.}
        
        \label{fig:stdist-t-lb}
    \end{figure}
    
        We reduce from $\tedge_N$. Alice and Bob construct a tournament $G = (V, E)$ with $n = N+5$ vertices (see \Cref{fig:stdist-t-lb}). Alice, based on her input, initially constructs the tournament $T$ with $N$ vertices. Let $u$ and $v$ be the vertices between which Bob needs to identify the orientation of the edge, and let $T'$ be the subgraph of $T$ defined by the vertex set $V - \{u, v\}$. Let $s$, $t$, $x_1$, $x_2$ and $x_3$ be the remaining vertices. Bob adds edges $(s, u)$, $(v, t)$ and $(t, s)$. He also adds edges $(s, x_1)$, $(x_1, x_2)$, $(x_2, x_3)$ and $(x_3, t)$. Next, he adds edges $(x_3, u)$, $(v, x_1)$ and $(x_3, x_1)$. Finally, for every vertex $z$ without an edge with $s$, he adds edge $(z, s)$, and for every vertex $z$ without an edge with $t$, he adds the edge $(t, z)$. Any remaining edges are oriented arbitrarily to complete the tournament. Note that $G$ is guaranteed to contain a path from $s$ to $t$, namely $s \to x_1 \to x_2 \to x_3 \to t$.
        
        Let us show that $G$ contains an edge oriented from $u$ to $v$ if and only if the distance from $s$ to $t$ is less than 4. If $G$ contains the edge $(u, v)$, we see that there is a path from $s$ to $t$ of length 3, namely, $s \to u \to v \to t$. Thus, the distance from $s$ to $t$ is less than 4. Next, assume that $G$ does not contain the edge $(u, v)$. By our construction of $G$, there is no edge from $s$ to $t$, so the distance can not be 1. Moreover, note that the only edges leaving $s$ are $(s, u)$ and $(s, x_1)$, and the only edges entering $t$ are $(v, t)$ and $(x_3, t)$. We know that there is no edge from $\{u, x_1\}$ to $t$, or from $s$ to $\{v, x_3\}$, so the distance can not be 2. Finally, there are no edges from $\{u, x_1\}$ to $\{v, x_3\}$, so the distance can not be 3. Thus, the distance from $s$ to $t$ must be at least 4.

        This implies that we can use a streaming algorithm for \stdistt on $n = N+5$ vertices to solve $\tedge_N$. Since $\tedge_N$ requires $\Omega(N^2)$ bits of one-way communication (see \Cref{prop:tedge-lb}), \stdistt must require $\Omega(n^2)$ space in a single pass.
    \end{proof}

    \section{Resolving the Streaming Complexities of Acyclicity Testing and Sink Finding}

\subsection{Testing Acyclicity of Tournaments}
   
  Acyclicity testing is closely related to FAS problems, since lower bounds on acyclicity testing give imply corresponding lower bounds on FAS approximation algorithms. In particular, a graph (or tournament) has FAS size 0 if and only if it is acyclic, so any multiplicative approximation to FAS size must also test acyclicity. In \cite{ChakrabartiGMV20}, the authors show that \acyct requires $\Omega(n/p)$ space in $p$ passes. We first give an algorithm matching this lower bound and settle the streaming complexity of \acyct. Then we present a simpler proof of the lower bound.

 \begin{restatable}{theorem}{acyctub}\label{thm:acyc-t-ub}
     Given an input tournament $T$, for any $p\geq 1$, there is a deterministic $p$-pass $\tO(n/p)$-space algorithm for detecting whether $T$ is acyclic or not. 
\end{restatable}

     \begin{proof}%[Proof of \Cref{thm:acyc-t-ub}]
        The algorithm is given by \Cref{alg:acyc}. 
        
         \begin{algorithm}[H]
      \caption{Acyclicity testing in $p$ passes and $\tO(n/p)$ space
        \label{alg:acyc}}
      \begin{algorithmic}[1]
      \Statex \textbf{Input}: Stream of edge insertions of an $n$-vertex tournament graph $G=(V,E)$ ($p$ passes)
      \Statex
        \Statex \underline{\textbf{Initialize}:}
        \State Partition $V$ into subsets $V_1, \ldots, V_p$ of size at most $\lceil n/p \rceil$ each
        \State $S \gets 0$

      \Statex 
      \Statex\underline{\textbf{Process} (pass $i$:)}
      \State $d(v) \gets 0$ for each $v\in V_i$
      \For{$(u, v) \in E$}
      \If{$u \in V_i$}
      \State $d(u) \gets d(u)+1$
      \EndIf
      \EndFor
      
    \Statex 
          \Statex \underline{\textbf{Post-processing} (pass $i$):}
          \For{$v \in V_i$}
          \State $S \gets S + d(v)^2$
          \EndFor

    \Statex
    \Statex \underline{\textbf{Post-processing} (final)}
    \State \If{$S = \frac{n(n-1)(2n-1)}{6}$}
    \State Return true
    \Else
    \State Return false
    \EndIf
          
      \end{algorithmic}
    \end{algorithm}

        It is clear that the algorithm requires $\tO(n/p)$ space, since every pass requires storing $n/p$ degree values of size at most $n$ and one sum of size at most $n^3$. The correctness of the algorithm follows directly from \cite{Moon1968TopicsOT} (Theorem 9), which states that a tournament graph $T$ is acyclic if and only if the sum of our degrees of the vertices is exactly $\frac{n(n-1)(2n-1)}{6}$.
    \end{proof}

The following theorem was proven in \cite{ChakrabartiGMV20}, but we provide a simpler proof here. 

    \begin{theorem}[\cite{ChakrabartiGMV20}]
 \label{acyct-lbsimple}
        Any randomized $p$-pass algorithm for acyclicty testing on tournaments requires $\Omega(n/p)$ space.
    \end{theorem}

\begin{proof}%[Proof of \Cref{acyct-lbsimple}]

    \begin{figure}[H]
        \centering
        
        \includegraphics[scale=0.8]{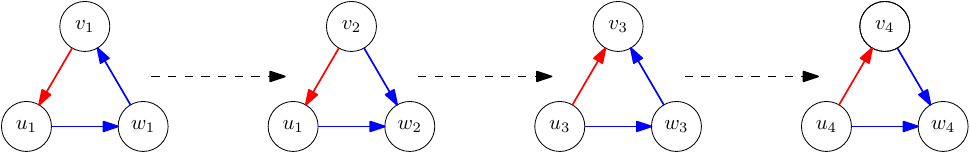}
        
        \caption{The tournament constructed in the lower bound for \acyct when Alice has the string $\bx = 1100$ and Bob has the string $\by = 1010$. Edges Alice adds are represented in red and edges Bob adds are represented in blue. Edges between clusters are represented by dashed arrows for clarity, and always go from left to right. Note that only the first cluster forms a cycle.}
        
        \label{fig:asyc-t-lb}
    \end{figure}
    
        We reduce from $\disj_N$. Alice and Bob construct a tournament $T = (V, E)$ with $n = 3N$ vertices, where vertices are labeled $\{u_1, v_1, w_1, \ldots, u_N, v_N, w_N\}$ (see \Cref{fig:asyc-t-lb}). Let $C_i$ refer to the subgraph induced by $u_i, v_i$ and $w_i$. Alice, based on her input $\bx$, adds the edge $(v_i, u_i)$ if $\bx_i = 1$, and edge $(u_i, v_i)$ if $\bx_i = 0$. Similarly, Bob, based on his input $\by$, adds edges $(u_i, w_i)$ and $(w_i, v_i)$ if $\by_i = 1$, and edges $(u_i, w_i)$ and $(v_i, w_i)$ if $\by_i = 0$. Bob also adds edges from every vertex in $C_i$ to every vertex in $C_j$ for all $1 \leq i < j \leq n$. This completes the construction of $T$. Note that $T$ is a tournament, since Alice and Bob together construct an edge between every pair of vertices within every cluster $C_i$ and between every pair of vertices in two different clusters $C_i$ and $C_j$.  

        Let us show that $T$ is acyclic if and only if $\bx \cap \by = \emptyset$. If $\bx$ and $\by$ are not disjoint, let $i$ be an index at which they intersect. Note that Alice adds edge $(v_i, u_i)$, and Bob adds edges $(u_i, w_i)$ and $(w_i, v_i)$. These 3 vertices therefore form a cycle in $T$, so $T$ is not acyclic. Thus, if $T$ is acyclic, $\bx \cap \by = \emptyset$. 
        
        Next, assume that $\bx \cap \by = \emptyset$. First, note that a cycle can not contain vertices from two distinct clusters $C_i$ and $C_j$, since if $i<j$ we see that there is no way to get from a vertex in $C_j$ to a vertex in $C_i$. It is also easy to check that if $\bx_i = 0$ or $\by_i = 0$, $C_i$ does not form a cycle; this can be seen in \Cref{fig:asyc-t-lb}, which illustrates all possible combinations. Thus, if $\bx \cap \by = \emptyset$, $T$ is acyclic. 

        This implies that we can use a streaming algorithm for \acyct on $n = 3N$ vertices to solve $\disj_N$. Since $\disj_N$ requires $\Omega(N)$ total bits of communication (see \ref{fact:disj-lb}), \acyct must require $\Omega(n/p)$ space in $p$ passes.
    \end{proof}

\subsection{Finding Sink (or Source) Nodes in DAGs}

Our final result is a lower bound on \sink. An $O(n/p)$ space upper bound can be obtained in $p$ passes by the following simple algorithm: partition the vertex set into $p$ groups. In the $i$th pass, keep a bit for each node in the $i$th group that keeps track of whether the node has an out-neighbor. If, after any of the passes, there is a vertex with no out-neighbor, the graph has a sink. If not, there is no sink. Our next result says that this is essentially the best we can do. 

\begin{restatable}{theorem}{sinkfind}\label{thm:sinkfind-lb}
     Any randomized $p$-pass algorithm that finds a sink node in a DAG requires $\Omega(n/p)$ space.
\end{restatable}

    \begin{proof}%[Proof of \Cref{thm:sinkfind-lb}]

    \begin{figure}[H]
        \centering
        
        \includegraphics{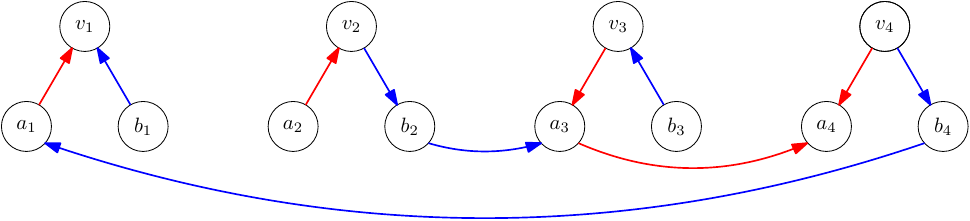}
        
        \caption{The graph constructed in the lower bound for \sink when Alice has the string $\bx = 1100$ and Bob has the string $\by = 1010$. Edges Alice adds are represented in red and edges Bob adds are represented in blue. Note that the sink $v_1$ corresponds to the shared first bit.}
        
        \label{fig:sink-lb}
    \end{figure}
    
        We reduce from $\setint_N$. Alice and Bob construct a graph $G = (V, E)$ with $n = 3N$ vertices, where vertices are labeled $\{v_1, a_1, b_1, \ldots, v_N, a_N, b_N\}$ (see \Cref{fig:sink-lb}). Let $j = i+1 \mod N$. Alice, based on her input $\bx$, adds edge $(a_i, v_i)$ if $\bx_i = 1$, and edge $(v_i, a_i)$ if $\bx_i = 0$. Similarly, Bob, based on his input $\by$, adds edge $(b_i, v_i)$ if $\by_i = 1$, and edge $(v_i, b_i)$ is $\by_i = 0$. If $\by_i = 0$, he also adds edge $(b_i, b_j)$ if $\by_j = 0$, and edge $(b_i, a_j)$ if $\by_j = 1$. This completes the construction of $G$. 
        
        Assume without loss of generality that $\bx$ and $\by$ intersect at index 1. We want to show that $v_1$ is the unique sink in $G$. By our construction, we know that Alice adds the edge $(a_1, v_1)$ and Bob adds the edge $(b_1, v_1)$. However, note that neither add any edge leaving $v_1$. Thus, $v_1$ is a sink in $G$. Moreover, note that $v_1$ is the only sink in $G$. This is because for all indices $j$, if $\bx_j = 1$ then $a_j$ has an outgoing edge to $v_j$, and if $\bx_j = 0$ then $a_j$ has an outgoing edge to either $a_{j+1 \mod N}$ or $b_{j+1 \mod N}$. Similarly, for all indices $j$, $b_j$ has an outgoing edge. Finally, since the intersection of $\bx$ and $\by$ is unique, we know that for all indices $j \neq 1$ either $\bx_j = 0$ or $\by_j = 0$, so $v_j$ must have an outgoing edge to either $a_j$ or $b_j$. Thus, $v_1$ is the unique sink in $G$. 

        For our construction to work, we must also show that $G$ is acyclic. To do so, we will construct a topological ordering on $G$. The ordering begins with section 1, which contains all vertices $a_i$ and $b_j$, where $2 \leq i, j \leq n$, such that $\bx_i = 1$ and $\by_j = 1$ respectively. These vertices are ordered in lexicographical order by index. Next, section 2 contains all vertices $v_i$, where $2 \leq i \leq n$, ordered by index. This is followed by section 3, namely all vertices $a_i$ and $b_j$, where $2 \leq i, j \leq n$, such that $\bx_i = 0$ and $\by_i = 0$ respectively, once again ordered by in lexicographical order by index. Finally, section 4 contains $a_1$, $b_1$ and $v_1$. For the example in~\ref{fig:sink-lb}, the ordering would be $\{a_2, b_3, v_2, v_3, v_4, b_2, a_3, a_4, b_4, a_1, b_1, v_1\}$. 
        
        Let us show that this is a topological ordering by verifying that every edge in our graph is a forward edge in the ordering. The first set of edges that Alice and Bob add are forward edges, since they either go from section 1 to section 2, from section 2 to section 3, or from $a_1$ and $b_1$ to $v_1$. The second set of edges they add always go from a vertex with index $i$ to a vertex with index $j = i+1 \mod N$, and unless $j = 1$, $i < j$ and both vertices must belong to section 3. Thus, all such edges where $j \neq 1$ are forward edges within section 3, and all such edges where $j = 1$ are forward edges from section 3 to section 4. Since the graph admits a topological ordering, it must be acyclic, completing the proof. 

        This implies that we can use a streaming algorithm for \sink on $n = 3N$ vertices to solve $\setint_N$. Since $\setint_N$ requires $\Omega(N)$ total bits of communication (see \ref{fact:disj-lb}, \sink must require $\Omega(n/p)$ space in $p$ passes.
    \end{proof}

\section{Implications of Our Results in Communication Complexity and Graph Theory}\label{sec:impli}

In this section, we note some important implications of our results in Communication Complexity and Combinatorial Graph Theory. These results might be of independent interest and relevant for future research on tournaments.

\subsection{Communication Complexity}
As a consequence of our results, the communication complexity of multiple tournament problems are resolved (up to polylogarithmic factors). For a graph problem $P$, we define its communication version as the following standard formulation in the two-party communication model: Alice and Bob hold graphs $G_A = (V, E_A)$ and $G_B = (V, E_B)$ respectively (with $E_A\cap E_B = \emptyset$), and they need to communicate with each other (in unrestricted rounds) so as to solve the problem $P$ on the union graph $G:= (V, E_A \sqcup E_B)$. The goal is to minimize the total number of bits communicated. Mande et al.~\cite{mande2024communication} studied the communication complexity of tournament problems such as finding \emph{kings} and maximum outdegree nodes. We add to this literature with the following results.

\begin{theorem}
    The deterministic and randomized communication complexity of each of the following problems is $\tTheta(n)$. 
    \begin{itemize}
        \item \sccdect 
        \item \reacht
        \item \sconnt
    \end{itemize}
\end{theorem}

\begin{proof}
    The communication upper bounds immediately follow from our streaming upper bounds given by \Cref{cor:sccdect,cor:reacht,cor:strongconnt} using the standard streaming-to-communication protocol simulation. The lower bounds for \reacht and \sconnt follow from the proofs of \Cref{thm:reach-lb} and  \Cref{thm:strongconn-lb} respectively. These imply the lower bound for \sccdect since solving it also solves the former two problems. Further, since our upper bounds are deterministic and lower bounds are randomized, we resolve both deterministic and randomized communication complexities of the problems.
\end{proof}

We also settle the communication complexity of \sink.

    \begin{theorem}
        The deterministic and randomized communication complexity of \sink is $\Theta(n)$.
    \end{theorem}

    \begin{proof}
        The randomized lower bound is given in the proof of \Cref{thm:sinkfind-lb}. The matching deterministic upper bound follows from the protocol where Alice sends Bob the $n$-bit characteristic vector of the set of sinks in her graph. Bob checks which of them is a sink in his graph as well.
    \end{proof}

\subsection{Combinatorial Graph Theory}

We get the following graph theoretic facts from our results. We believe that that they are significant on their own and should find further applications in combinatorial analysis and algorithm design for tournaments.

\begin{fact}\label{thm:combone}
    For two nodes $s,t$ in a tournament $T$, if $\din(s)\leq \din(t)$, then there is a directed path from $s$ to $t$ in $T$.
\end{fact}

\begin{proof}
    If there is no path from $s$ to $t$, then the SCC containing $t$ appears before the SCC containing $s$ in $T_\scc$. Then, there is an SCC-cut $(U, V\setminus U)$ such that $t\in U$ and $s\in V\setminus U$. By \Cref{lem:scc-inc-deg}, $\din(t)< \din(s)$, which is a contradiction. Hence, there must be a path from $s$ to $t$.
\end{proof}

\begin{fact}
    The SCCs of a tournament are completely determined by the set of its indegrees (or outdegrees). 
\end{fact}

\begin{proof}
    Follows from \Cref{alg:sccgraph} and its correctness.
\end{proof}

\begin{fact}
    A tournament $T$ is strongly connected if and only if there is a path from $v_{\text{max}}$ to $v_{\text{min}}$, where  $v_{\text{max}}$ and $v_{\text{min}}$ are its max-indegree and min-indegree nodes respectively.
\end{fact}

\begin{proof}
    The ``only if'' direction follows from the definition of strong connectivity. For the ``if'' direction, consider any two nodes $s$ and $t$ in $T$, where $\din(s)\leq \din(t)$ WLOG. By \Cref{thm:combone}, there is a directed path from $s$ to $t$. Again, since $\din(t)\leq \din(v_{\text{max}})$ and $\din(v_{\text{min}})\leq \din(s)$, there is a path from $t$ to $v_{\text{max}}$, from $v_{\text{max}}$ to $v_{\text{min}}$, and from $v_{\text{min}}$ to $s$, i.e., from $t$ to $s$.  
\end{proof}
    
    \section*{Acknowledgements}

    We are extremely grateful to Srijon Mukherjee for letting us know about the IOITC question on determining strong connectivity of a tournament from its set of indegrees, solving and building on which we obtained our results for SCC decomposition. We also thank Sepehr Assadi, Amit Chakrabarti, and Madhu Sudan for helpful discussions. Finally, we thank the organizers of the DIMACS REU program for making this collaboration possible.  

\bibliography{refs}

\newcommand{\etalchar}[1]{$^{#1}$}
\begin{thebibliography}{CKP{\etalchar{+}}21b}

\bibitem[Abl96]{Ablayev96}
Farid Ablayev.
\newblock Lower bounds for one-way probabilistic communication complexity and their application to space complexity.
\newblock {\em Theoretical Computer Science}, 175(2):139--159, 1996.

\bibitem[ACG{\etalchar{+}}15]{AhnCGMW15}
KookJin Ahn, Graham Cormode, Sudipto Guha, Andrew McGregor, and Anthony Wirth.
\newblock Correlation clustering in data streams.
\newblock In {\em Proceedings of the 32nd International Conference on Machine Learning}, volume~37 of {\em Proceedings of Machine Learning Research}, pages 2237--2246. PMLR, 2015.

\bibitem[AJJ{\etalchar{+}}22]{AssadiJJST22}
Sepehr Assadi, Arun Jambulapati, Yujia Jin, Aaron Sidford, and Kevin Tian.
\newblock Semi-streaming bipartite matching in fewer passes and optimal space.
\newblock In {\em Proceedings of the 2022 {ACM-SIAM} Symposium on Discrete Algorithms, {SODA} 2022, Virtual Conference / Alexandria, VA, USA, January 9 - 12, 2022}, pages 627--669. {SIAM}, 2022.

\bibitem[AKSY20]{AssadiKSY20}
Sepehr Assadi, Gillat Kol, Raghuvansh~R. Saxena, and Huacheng Yu.
\newblock Multi-pass graph streaming lower bounds for cycle counting, max-cut, matching size, and other problems.
\newblock In {\em 61st {IEEE} Annual Symposium on Foundations of Computer Science, {FOCS} 2020, Durham, NC, USA, November 16-19, 2020}, pages 354--364. {IEEE}, 2020.

\bibitem[AR20]{AssadiR20}
Sepehr Assadi and Ran Raz.
\newblock Near-quadratic lower bounds for two-pass graph streaming algorithms.
\newblock In {\em 61st {IEEE} Annual Symposium on Foundations of Computer Science, {FOCS} 2020, Durham, NC, USA, November 16-19, 2020}, pages 342--353. {IEEE}, 2020.

\bibitem[Ass23]{AssadiMultipassSurvey}
Sepehr Assadi.
\newblock Recent advances in multi-pass graph streaming lower bounds.
\newblock {\em ACM SIGACT News}, 54(3):48–75, Sep 2023.

\bibitem[BCK{\etalchar{+}}14]{brodyckwy14}
Joshua Brody, Amit Chakrabarti, Ranganath Kondapally, David~P. Woodruff, and Grigory Yaroslavtsev.
\newblock Certifying equality with limited interaction.
\newblock In {\em Proc. 18th International Workshop on Randomization and Approximation Techniques in Computer Science}, pages 545--581, 2014.

\bibitem[BCMT23]{BehnezhadCMT23}
Soheil Behnezhad, Moses Charikar, Weiyun Ma, and Li{-}Yang Tan.
\newblock Single-pass streaming algorithms for correlation clustering.
\newblock In Nikhil Bansal and Viswanath Nagarajan, editors, {\em Proceedings of the 2023 {ACM-SIAM} Symposium on Discrete Algorithms, {SODA} 2023, Florence, Italy, January 22-25, 2023}, pages 819--849. {SIAM}, 2023.

\bibitem[BdW02]{BuhrmanW02}
Harry Buhrman and Ronald de~Wolf.
\newblock Complexity measures and decision tree complexity: a survey.
\newblock {\em Theor. Comput. Sci.}, 288(1):21--43, 2002.

\bibitem[BGNR98]{BarYehudaGNR98}
Reuven Bar{-}Yehuda, Dan Geiger, Joseph~(Seffi) Naor, and Ron~M. Roth.
\newblock Approximation algorithms for the feedback vertex set problem with applications to constraint satisfaction and bayesian inference.
\newblock {\em SIAM Journal on Computing}, 27(4):942--959, 1998.

\bibitem[BJW22]{BawejaJW22}
Anubhav Baweja, Justin Jia, and David~P. Woodruff.
\newblock An efficient semi-streaming {PTAS} for tournament feedback arc set with few passes.
\newblock In {\em 13th Innovations in Theoretical Computer Science Conference, {ITCS} 2022, January 31 - February 3, 2022, Berkeley, CA, {USA}}, volume 215 of {\em LIPIcs}, pages 16:1--16:23. Schloss Dagstuhl - Leibniz-Zentrum f{\"{u}}r Informatik, 2022.

\bibitem[BW78]{beineke1978selected}
L.W. Beineke and R.J. Wilson.
\newblock {\em Selected Topics in Graph Theory}.
\newblock Number v. 1-3 in Selected Topics in Graph Theory. Academic Press, 1978.

\bibitem[Cam59]{camion1959chemins}
Paul Camion.
\newblock Chemins et circuits hamiltoniens des graphes complets.
\newblock {\em COMPTES RENDUS HEBDOMADAIRES DES SEANCES DE L ACADEMIE DES SCIENCES}, 249(21):2151--2152, 1959.

\bibitem[CCGT14]{ChakrabartiCGT14}
Amit Chakrabarti, Graham Cormode, Navin Goyal, and Justin Thaler.
\newblock Annotations for sparse data streams.
\newblock In {\em Proc. 25th Annual ACM-SIAM Symposium on Discrete Algorithms}, pages 687--706, 2014.

\bibitem[CFR10]{CoppersmithFR10}
Don Coppersmith, Lisa Fleischer, and Atri Rudra.
\newblock Ordering by weighted number of wins gives a good ranking for weighted tournaments.
\newblock {\em {ACM} Trans. Algorithms}, 6(3):55:1--55:13, 2010.

\bibitem[CGMV20]{ChakrabartiGMV20}
Amit Chakrabarti, Prantar Ghosh, Andrew McGregor, and Sofya Vorotnikova.
\newblock Vertex ordering problems in directed graph streams.
\newblock In {\em Proceedings of the 2020 {ACM-SIAM} Symposium on Discrete Algorithms, {SODA} 2020, Salt Lake City, UT, USA, January 5-8, 2020}, pages 1786--1802. {SIAM}, 2020.

\bibitem[CKP{\etalchar{+}}21a]{ChenKPSSY21-lognpass}
Lijie Chen, Gillat Kol, Dmitry Paramonov, Raghuvansh~R. Saxena, Zhao Song, and Huacheng Yu.
\newblock Almost optimal super-constant-pass streaming lower bounds for reachability.
\newblock In {\em {STOC} '21: 53rd Annual {ACM} {SIGACT} Symposium on Theory of Computing, Virtual Event, Italy, June 21-25, 2021}, pages 570--583. {ACM}, 2021.

\bibitem[CKP{\etalchar{+}}21b]{ChenKPSSY21-randomwalk}
Lijie Chen, Gillat Kol, Dmitry Paramonov, Raghuvansh~R. Saxena, Zhao Song, and Huacheng Yu.
\newblock Near-optimal two-pass streaming algorithm for sampling random walks over directed graphs.
\newblock In {\em 48th International Colloquium on Automata, Languages, and Programming, {ICALP} 2021, July 12-16, 2021, Glasgow, Scotland (Virtual Conference)}, volume 198 of {\em LIPIcs}, pages 52:1--52:19. Schloss Dagstuhl - Leibniz-Zentrum f{\"{u}}r Informatik, 2021.

\bibitem[CSWY01]{chakrabartiswy01}
Amit Chakrabarti, Yaoyun Shi, Anthony Wirth, and Andrew~C. Yao.
\newblock Informational complexity and the direct sum problem for simultaneous message complexity.
\newblock In {\em Proc. 42nd Annual IEEE Symposium on Foundations of Computer Science}, pages 270--278, 2001.

\bibitem[CTY07]{CharbitTA07}
Pierre Charbit, Stéphan Thomassé, and Anders Yeo.
\newblock The minimum feedback arc set problem is np-hard for tournaments.
\newblock {\em Combinatorics, Probability \& Computing}, 16:1--4, 01 2007.

\bibitem[DV13]{DasV13}
Abhik~Kumar Das and Sriram Vishwanath.
\newblock On finite alphabet compressive sensing.
\newblock In {\em 2013 IEEE International Conference on Acoustics, Speech and Signal Processing}, pages 5890--5894, 2013.

\bibitem[Elk17]{Elkin17}
Michael Elkin.
\newblock Distributed exact shortest paths in sublinear time.
\newblock In {\em Proc. 49th Annual ACM Symposium on the Theory of Computing}, pages 757--770, 2017.

\bibitem[FKM{\etalchar{+}}05]{feigenbaumkmsz05}
Joan Feigenbaum, Sampath Kannan, Andrew McGregor, Siddharth Suri, and Jian Zhang.
\newblock On graph problems in a semi-streaming model.
\newblock {\em Theor. Comput. Sci.}, 348(2--3):207--216, 2005.
\newblock Preliminary version in \em Proc. 31st International Colloquium on Automata, Languages and Programming\em\/, pages 531--543, 2004.

\bibitem[Gas98]{GassTournaments}
S~I Gass.
\newblock Tournaments, transitivity and pairwise comparison matrices.
\newblock {\em Journal of the Operational Research Society}, 49(6):616--624, 1998.

\bibitem[GI10]{gilbertI10}
Anna~C. Gilbert and Piotr Indyk.
\newblock Sparse recovery using sparse matrices.
\newblock {\em Proceedings of the {IEEE}}, 98(6):937--947, 2010.

\bibitem[GLQ{\etalchar{+}}08]{GengLQALS08}
Xiubo Geng, Tie-Yan Liu, Tao Qin, Andrew Arnold, Hang Li, and Heung-Yeung Shum.
\newblock Query dependent ranking using k-nearest neighbor.
\newblock In {\em Proceedings of the 31st Annual International ACM SIGIR Conference on Research and Development in Information Retrieval}, SIGIR '08, page 115–122. Association for Computing Machinery, 2008.

\bibitem[GO13]{GuruswamiO13}
Venkatesan Guruswami and Krzysztof Onak.
\newblock Superlinear lower bounds for multipass graph processing.
\newblock In {\em Proceedings of the 28th Conference on Computational Complexity, {CCC} 2013, K.lo Alto, California, USA, 5-7 June, 2013}, pages 287--298. {IEEE} Computer Society, 2013.

\bibitem[GVV17]{GuruswamiVV17}
Venkatesan Guruswami, Ameya Velingker, and Santhoshini Velusamy.
\newblock Streaming complexity of approximating max 2csp and max acyclic subgraph.
\newblock In {\em Approximation, Randomization, and Combinatorial Optimization. Algorithms and Techniques, {APPROX/RANDOM} 2017, August 16-18, 2017, Berkeley, CA, {USA}}, volume~81 of {\em LIPIcs}, pages 8:1--8:19. Schloss Dagstuhl - Leibniz-Zentrum f{\"{u}}r Informatik, 2017.

\bibitem[HRR99]{HRR98}
Monika~R. Henzinger, Prabhakar Raghavan, and Sridhar Rajagopalan.
\newblock Computing on data streams.
\newblock {\em External memory algorithms}, pages 107--118, 1999.

\bibitem[Jin19]{Jin19}
Ce~Jin.
\newblock Simulating random walks on graphs in the streaming model.
\newblock In {\em 10th Innovations in Theoretical Computer Science Conference, {ITCS} 2019, January 10-12, 2019, San Diego, California, {USA}}, pages 46:1--46:15, 2019.

\bibitem[Kar72]{Karp72}
Richard Karp.
\newblock Reducibility among combinatorial problems.
\newblock {\em Complexity of Computer Computations}, 40:85--103, 1972.

\bibitem[KS07]{KenyonS07}
Claire Kenyon{-}Mathieu and Warren Schudy.
\newblock How to rank with few errors.
\newblock In {\em Proceedings of the Thirty-Ninth Annual ACM Symposium on Theory of Computing}, STOC '07, page 95–103. Association for Computing Machinery, 2007.

\bibitem[LS11]{lauraS11}
Luigi Laura and Federico Santaroni.
\newblock Computing strongly connected components in the streaming model.
\newblock In {\em Theory and Practice of Algorithms in (Computer) Systems}, pages 193--205. Springer Berlin Heidelberg, 2011.

\bibitem[McG14]{McGregor14}
Andrew McGregor.
\newblock Graph stream algorithms: a survey.
\newblock {\em {SIGMOD} Record}, 43(1):9--20, 2014.

\bibitem[Moo68]{Moon1968TopicsOT}
John~W. Moon.
\newblock {\em Topics on tournaments}.
\newblock Athena series; selected topics in mathematics. Holt, Rinehart and Winston, 1968.

\bibitem[MPSS24]{mande2024communication}
Nikhil~S. Mande, Manaswi Paraashar, Swagato Sanyal, and Nitin Saurabh.
\newblock On the communication complexity of finding a king in a tournament.
\newblock {\em ArXiv preprint}, abs/2402.14751, 2024.

\bibitem[MV17]{McGregorV17}
Andrew McGregor and Hoa~T. Vu.
\newblock Better streaming algorithms for the maximum coverage problem.
\newblock In {\em 20th International Conference on Database Theory, {ICDT} 2017, March 21-24, 2017, Venice, Italy}, volume~68 of {\em LIPIcs}, pages 22:1--22:18. Schloss Dagstuhl - Leibniz-Zentrum f{\"{u}}r Informatik, 2017.

\bibitem[Raz92]{Razborov92}
Alexander Razborov.
\newblock On the distributional complexity of disjointness.
\newblock {\em Theor. Comput. Sci.}, 106(2):385--390, 1992.
\newblock Preliminary version in \em Proc. 17th International Colloquium on Automata, Languages and Programming\em\/, pages 249--253, 1990.

\bibitem[SGP11]{SarmaGP11}
Atish~Das Sarma, Sreenivas Gollapudi, and Rina Panigrahy.
\newblock Estimating pagerank on graph streams.
\newblock {\em J. ACM}, 58(3):13, 2011.

\bibitem[Sor88]{SorokerHamCycle}
Danny Soroker.
\newblock Fast parallel algorithms for finding hamiltonian paths and cycles in a tournament.
\newblock {\em J. Algorithms}, 9(2):276–286, June 1988.

\bibitem[SST16]{SimpsonST16}
Michael Simpson, Venkatesh Srinivasan, and Alex Thomo.
\newblock Efficient computation of feedback arc set at web-scale.
\newblock {\em Proc. VLDB Endow.}, 10(3):133–144, 2016.

\bibitem[Yao79]{yao79}
Andrew~C. Yao.
\newblock Some complexity questions related to distributive computing.
\newblock In {\em Proc. 11th Annual ACM Symposium on the Theory of Computing}, pages 209--213, 1979.

\end{thebibliography}
\bibliographystyle{alpha}

    \end{document}